\documentclass[lettersize,journal]{IEEEtran}

\usepackage{amsmath,amsthm,amssymb,amsfonts}
\usepackage{graphicx}
\usepackage{booktabs}
\usepackage{array}
\usepackage{url}
\usepackage{cite}
\usepackage{multirow}
\usepackage{xcolor}
\usepackage{tikz}
\usetikzlibrary{arrows.meta, positioning}
\usepackage[caption=false,font=normalsize,labelfont=sf,textfont=sf]{subfig}
\usepackage{stfloats}
\usepackage{enumitem}
\definecolor{copenblue}{RGB}{0,102,204}
\definecolor{copenlight}{RGB}{173,216,230}
\definecolor{copendark}{RGB}{0,51,102}

\newtheorem{lemma}{Lemma}
\newtheorem{corollary}{Corollary}
\newtheorem{definition}{Definition}
\newtheorem{proposition}{Proposition}
\newtheorem{assumption}{Assumption}
\newtheorem{remark}{Remark}

\newcommand{\E}{\mathbb{E}}

\newcommand{\R}{\mathbb{R}}

\begin{document}

\title{Structural Decomposability of Encrypted Traffic\\ Side-Channel Leakage}

\author{Guangjie~Liu,~Guang~Cheng,~and~Weiwei~Liu%
\thanks{The authors are with the School of Cyber Science and Engineering,
Southeast University, Nanjing 210096, China
(e-mail: \{gjliu, gcheng, wwliu\}@seu.edu.cn).}%
\thanks{Manuscript received \today.}}

\markboth{IEEE Transactions on Dependable and Secure Computing}%
{Liu \MakeLowercase{\textit{et al.}}: Structural Decomposability of Encrypted Traffic Side-Channel Leakage}

\maketitle

\begin{abstract}
Existing side-channel theories treat leakage as a holistic quantity
$I(X;Y)$, lacking a formal characterization of its internal structure.
This paper systematically investigates the \emph{structural
decomposability} of encrypted traffic side-channel leakage.
Based on the structural causal model (SCM)
$X\!\to\!Y_{\mathrm{size}}\!\to\!Y_{\mathrm{dir}}\!\to\!Y_{\mathrm{time}}$
and the mutual-information chain rule, total leakage is precisely
decomposed into three sequential incremental terms for packet size,
direction, and timing.
Defense operations are formalized through the mechanism-replacement
semantics of a defense strategy variable~$D$.
Coupling information $C_{\mathrm{size,dir}}=I(Y_{\mathrm{size}};Y_{\mathrm{dir}}\!\mid\! X)$
is introduced to measure inter-dimensional statistical dependence,
and the Markov residual serves as a testable sufficient condition for
a single-dimension defense to sever subsequent leakage.
Causal efficacy~$\eta_d$ is defined to quantify per-unit-cost leakage
suppression efficiency.
Under a differentiable parametric family and a small-perturbation
assumption, a Fisher-geometric approximation
$I(X;Y)\approx\frac{1}{2\ln 2}\mathrm{Tr}(G\Sigma_\theta)$
is established.
Empirical validation on the Wang dataset (95 websites) shows:
$Y_{\mathrm{dir}}$ dominates leakage under the undefended baseline
(0.637\,bits), while Tor's fixed 512-byte cells cause
$Y_{\mathrm{size}}$ to degenerate ($I(X;Y_{\mathrm{size}})\!\approx\!0$);
FRONT suppresses the direction term by 63\%, yet the Markov residual
remains 0.021\,bits (bootstrap 95\% CI $[0.016,0.027]$ excludes zero),
indicating that single-dimension direction defense does not satisfy the
sufficient condition for severing timing leakage;
$\eta_{\mathrm{dir}}=0.97$ and $\eta_{\mathrm{time}}\approx 0$
quantitatively confirm FRONT's design intent of replacing only the
direction generation mechanism.
The paper unifies chain decomposition, coupling measurement, and
Fisher-geometric analysis to provide a computable, structured leakage
accounting method and a theoretical foundation for multi-dimensional
joint defense design.
\end{abstract}


\section{Introduction}

\IEEEPARstart{W}{ith} the widespread deployment of encryption protocols
such as TLS~1.3 and QUIC, the payload content of modern network
communications is strongly protected by cryptography.
Nevertheless, side-channel analysis based on metadata features---packet
lengths, inter-arrival times, and directions---continues to achieve high
accuracy on multiple tasks: website fingerprinting reaches
91--98\% accuracy in closed-world settings~\cite{wang2014effective,sirinam2018deep},
and many works on application identification and malicious traffic
detection report over 90\% accuracy~\cite{shen2021graphdapp,lin2022etbert}.
Encryption protects payload content but cannot conceal communication
metadata: packet lengths, timestamps, and directions are indispensable
for network routing and transport control, forming the physical basis of
side-channel leakage.
Under a fixed observational definition and non-degenerate distributional
differences, the mutual information $I(X;Y)$ is necessarily positive,
establishing the inevitability of side-channel
leakage~\cite{liu2026inevitability}.

Since leakage is unavoidable, defense becomes the central issue.
Most existing defenses are designed for a single dimension: Tamaraw
eliminates length differences via fixed-length padding; WTF-PAD
obfuscates timing patterns via random delay; yet both have been broken by
deep-learning attacks exploiting multi-dimensional
features~\cite{sirinam2018deep}.
The systematic review by Mathews~et~al.~\cite{mathews2023sok}
(SoK, IEEE S\&P~2023) evaluated nine prevalent defenses and found that
most failed to provide their claimed protection.
An adversary can simply shift to an unprotected dimension and recover
high accuracy---the failure of single-dimension defenses appears to be a
universal pattern, yet it has remained only an empirical observation
without a formal structural explanation.

The root cause is that existing side-channel theories treat leakage as a
holistic quantity $I(X;Y)$, lacking a formal characterization of its
internal structure.
Total leakage is an opaque whole: it cannot account for how much each of
size, direction, and timing contributes, nor can it characterize the
statistical dependences among dimensions---which is precisely the key to
understanding why single-dimension defenses fail and how effective joint
defenses should be designed.
To answer these questions we must decompose $I(X;Y)$ into structured
components for each dimension, establishing the
\textbf{structural decomposability} of side-channel leakage.

The core tool of structural decomposability is the mutual-information
chain rule.
For observed features
$Y=(Y_{\mathrm{size}},Y_{\mathrm{dir}},Y_{\mathrm{time}})$,
the chain rule yields the exact identity
\begin{align}\label{eq:intro-chain}
I(X;Y) &= I(X;Y_{\mathrm{size}})
         + I(X;Y_{\mathrm{dir}}\mid Y_{\mathrm{size}})\notag\\
        &\quad + I(X;Y_{\mathrm{time}} \mid Y_{\mathrm{size}},Y_{\mathrm{dir}}).
\end{align}
The sum of the three terms equals the total leakage exactly, with no
omissions and no double-counting.
Each term has a clear physical interpretation:
$I(X;Y_{\mathrm{size}})$ is the leakage when the adversary observes only
packet-size features;
$I(X;Y_{\mathrm{dir}}\mid Y_{\mathrm{size}})$ is the \emph{additional}
leakage from observing direction features given size is already known;
$I(X;Y_{\mathrm{time}}\mid Y_{\mathrm{size}},Y_{\mathrm{dir}})$ is the
additional leakage from timing patterns given that both size and direction
are already known.

However, the chain identity alone is a basic property of information
theory; translating it into a defense analysis tool requires two
theoretical pillars.
The first is the \textbf{structural causal model (SCM)}, which specifies
``which generation mechanism a defense operation alters.''
We use Pearl's SCM framework~\cite{pearl2009causality} to build the
causal chain
$X\!\to\!Y_{\mathrm{size}}\!\to\!Y_{\mathrm{dir}}\!\to\!Y_{\mathrm{time}}$
and introduce the defense strategy variable~$D$ to describe the
mechanism-replacement effect of a defense---we do not pursue causal
effect identification via do-calculus, but rather precisely specify
``padding replaces the parameters of $f_{\mathrm{size}}$, direction
obfuscation replaces the parameters of $f_{\mathrm{dir}}$, and timing
delay replaces the parameters of $f_{\mathrm{time}}$.''
This creates a precise correspondence between ``defense design'' and
``chain component suppression,'' with testable propositions formulated in
terms of information quantities (mutual information, conditional mutual
information, coupling information, Markov residual) rather than causal
identification.
The second pillar is \textbf{information geometry}, used to characterize
the local separability change due to a defense.
We are interested not in the performance curve of a specific classifier,
but in the local curvature of distributional separability:
under a regular parametric family and a small-perturbation assumption,
the local variation in mutual information can be approximated as
$\frac{1}{2\ln 2}\mathrm{Tr}(G\Sigma_\theta)$, where $G$ is the Fisher
information matrix and $\Sigma_\theta$ is the parameter covariance matrix.
This geometric approximation converts design questions such as ``which
dimension is more compressible'' and ``does a single-dimension replacement
alter the spectral properties of other dimensions'' into analyzable
geometric propositions, and gives causal efficacy $\eta_d$---the leakage
suppression per unit defense cost---a computable form.

This paper does not attempt to prove in full generality that ``any
single-dimension defense must fail.''
The output is a structured decomposition and testable criteria: they tell
you ``through which information channel leakage escapes'' and ``whether a
given sufficient condition is satisfied.''
Empirical validation on the Wang dataset (95 websites) shows that
$Y_{\mathrm{dir}}$ dominates leakage under the undefended baseline
(0.637\,bits), while Tor's fixed 512-byte cells cause
$Y_{\mathrm{size}}$ to degenerate ($I(X;Y_{\mathrm{size}})\approx 0$);
FRONT suppresses the direction term by 63\%, but the Markov residual is
0.021\,bits (bootstrap 95\% CI $[0.016,0.027]$ excludes zero), indicating
that replacing only the direction mechanism does not satisfy the sufficient
condition for severing timing leakage;
causal efficacy $\eta_{\mathrm{dir}}=0.97$ and
$\eta_{\mathrm{time}}\approx 0$, quantitatively confirming FRONT's design
intent of targeting only the direction channel.

This paper is positioned as follows: building on results establishing the
inevitability of leakage ($I(X;Y)>0$)~\cite{liu2026inevitability}, we
systematically study the structural decomposability of leakage to provide a
computable, structured leakage accounting method and theoretical foundation
for ``where leakage comes from, how each dimension contributes, and when a
single-dimension defense is sufficient.''
In the empirical section, to accommodate the properties of publicly
available data (the Wang dataset comes from Tor, whose fixed 512-byte
cells cause $Y_{\mathrm{size}}$ to degenerate), we focus on the
direction dimension $Y_{\mathrm{dir}}$ and timing dimension
$Y_{\mathrm{time}}$ via proxy statistics
$(\hat{Y}_{\mathrm{dir},1}, \hat{Y}_{\mathrm{time}},
\hat{Y}_{\mathrm{dir},2})$ for a computability demonstration; this
demonstration serves framework validation and reflects the structural
characteristics of the Tor scenario.

The paper is organized as follows. Section~\ref{sec:related} surveys
related work. Section~\ref{sec:scm} establishes the causal decoupling
model. Section~\ref{sec:ig} introduces the information-geometry framework.
Section~\ref{sec:leakage} develops the leakage decomposition and causal
efficacy. Section~\ref{sec:empirical} presents empirical validation.
Section~\ref{sec:discussion} discusses theoretical boundaries and defense
implications. Section~\ref{sec:conclusion} concludes.

\section{Related Work}\label{sec:related}

\subsection{Attack Side: Multi-Dimensional Feature Fusion and Long-Range Dependence}

From the attack perspective, the core of encrypted-traffic side-channel
analysis is not any particular feature or classifier, but rather the
adversary's progressive expansion of observable feature sets and
exploitation of their conditional incremental information: from
single-dimensional statistics to direction/timing sequences, and then to
multi-scale long-range modeling.
The improvement in attack performance reflects a systematic extraction of
``residual separability brought by newly observed dimensions.''
This viewpoint provides an attack-theoretic rationale for the chain
accounting of total leakage by $(Y_{\mathrm{size}},
Y_{\mathrm{dir}}, Y_{\mathrm{time}})$ in this paper.

\paragraph{Multi-dimensional Feature Fusion}
Wang~et~al.~\cite{wang2014effective} found in an open-world setting that
the total number of incoming packets is the single most discriminative
feature; Hayes and Danezis~\cite{hayes2016kfingerprinting} maintained
high identification rates even with a monitored set of 100,000 unmonitored
pages, showing that timing features of long flows carry rich semantic
information.
Sirinam~et~al.~\cite{sirinam2018deep} achieved high accuracy on Tor using
end-to-end sequential representations (direction/timing and their
combination) in Deep Fingerprinting, a representative and widely reproduced
result against WTF-PAD (more robust representation attacks have since
appeared~\cite{shen2023subverting});
Reed and Kranch~\cite{reed2017netflix} demonstrated the strong
discriminative power of long-term timing patterns by identifying HTTPS
video streams from extended observations.
More recently, multi-dimensional fusion has been further systematized:
Qu~et~al.~\cite{qu2023hierarchical} proposed a three-level hierarchical
deep-learning framework (packet $\to$ flow $\to$ trace) showing the
complementarity of size, direction, and timing;
Jin~et~al.~\cite{jin2023transformer} extended the Transformer architecture
to multi-tab fingerprinting, demonstrating that long flows retain
sufficient distinguishable structure in mixed traffic.

\paragraph{Long-Range Dependence and Early-Stage Identification}
These results collectively indicate that timing and direction accumulate
information continuously over long ranges.
Deng~et~al.~\cite{deng2024earlystage} quantified the accuracy gap between
early (few-packet) and full-trace observation, providing evidence of
significant information increments between early and long-range observations.
Cherubin~et~al.~\cite{cherubin2022online} measured structural trade-offs
between monitored-set size and feature-dimension utilization on real Tor
exit relays.
Siby~et~al.~\cite{siby2023quic} provided packet-sampling-rate versus
attack-accuracy curves for QUIC, showing that multi-dimensional features
remain robust under short observation windows.
Mei~et~al.~\cite{mei2025high} addressed the base-rate fallacy of extremely
low anonymous-traffic fractions in open-world deployments, achieving more
than an eightfold accuracy improvement under a
Tor/non-Tor ratio of 1:1000.

From a theoretical perspective, Liu~et~al.~\cite{liu2026inevitability}
proved the inevitability of encrypted-traffic side-channel leakage
(Side-Channel Existence Theorem) using information theory, showing that
$I(X;Y)>0$ is a structural result that cannot be eliminated by encryption.
However, that work treats leakage as a holistic quantity $I(X;Y)$ without
addressing its internal structure---which decomposable structural components
contribute to total leakage $I(X;Y)$?
These attack studies exploit multi-dimensionality in practice but lack
formal theoretical support: why can these different-dimensional features be
used independently or jointly for leakage inference?
This is the direct motivation for introducing chain decomposition and
coupling information in this paper.

\subsection{Defense Side: Limitations of Single-Dimension Defenses}

Existing defense strategies are diverse in implementation, but many are
designed with \emph{a single atomic dimension as the direct intervention
target}, obfuscating specific metadata features through padding, delay,
or shaping.
BuFLO~\cite{dyer2012peek} fixes the sending rate and can theoretically
eliminate differences, but at over 100\% bandwidth overhead.
Tamaraw~\cite{cai2014systematic} eliminates size differences via
fixed-length padding and establishes a theoretical lower bound on bandwidth
for any defense.
Walkie-Talkie~\cite{wang2017walkie} uses ``half-duplex'' burst shaping to
further reduce overhead.
WTF-PAD~\cite{juarez2016wtfpad} obfuscates timing patterns via random
delay, but has been broken by Deep Fingerprinting~\cite{sirinam2018deep}.
NetShaper~\cite{sabzi2024netshaper} is the first to model network
side-channel mitigation as an $(\varepsilon,\delta)$-differential privacy
problem, providing quantifiable privacy--overhead trade-offs.
Palette~\cite{shen2024palette} clusters websites with similar traffic
patterns and normalizes them to a unified template.
Huang~et~al.~\cite{huang2025wfa2d} proposed the asymmetric adversarial
defense WF-A2D, achieving 97\% defense success rate against seven attacks
on HTTP/3-QUIC with less than 2\% bandwidth overhead.
These strategies share a common property: they strongly obfuscate one
dimension (e.g., size or timing), yet leakage in other dimensions (e.g.,
timing or direction patterns) may still be exploited.

At the mechanism level, most existing defenses fall into three categories:
(i) \emph{regularization/shaping}, which reduces inter-class differences
in one dimension (e.g., Tamaraw's fixed length, BuFLO's fixed rate,
Walkie-Talkie's half-duplex burst shaping);
(ii) \emph{randomization/injection}, which alters the identifiability of
the observation channel (e.g., WTF-PAD's random delay, FRONT's dummy
packet injection);
(iii) \emph{splitting}, which reduces the observation resolution available
to the adversary (e.g., TrafficSliver's multi-path routing).
The common difficulty of all three is that they typically target a single
dimension directly, but adversaries can exploit the conditional incremental
information of the remaining dimensions, producing the empirical phenomenon
of ``single-dimension effectiveness, but overall leakage persists.''

The systematic review by Mathews~et~al.~\cite{mathews2023sok} revealed
the prevalence of single-dimension defense failure:
reexamining nine purportedly efficient WF defenses with state-of-the-art
DL attacks, most failed to provide their claimed protection.
Tamaraw (fixed-length padding) was broken because adversaries could exploit
timing features; WTF-PAD (random delay) was broken because adversaries could
exploit size and direction patterns.
Subsequent work further showed that CNN representations robust to traffic
shaping can break multiple defenses, demonstrating the fundamental
fragility of heuristic shaping under adaptive attacks.
The empirical conclusion of that SoK is that the phenomenon is widespread;
the task of this paper is to provide ``explanation and accounting''---connecting
the failure phenomenon to measurable structural quantities such as coupling
information and Markov residuals via a leakage decomposition framework.
We emphasize that this paper does not use attack-accuracy curves as the
primary output metric, but rather uses chain mutual information components,
coupling information, Markov residuals, and causal efficacy as structural
outputs; the empirical section demonstrates that these quantities are
computable and aligned with the defense mechanism family.

\subsection{Theoretical Tools: Causal Inference and Information Geometry}

Pearl's structural causal model (SCM)~\cite{pearl2009causality} provides
a formal language for ``mechanism replacement/mechanism selection'': a
defender's choice of strategy value $D=d$ can represent the adoption of a
specific defense mechanism, thereby replacing the parameter family of the
corresponding generation mechanism at the structural-equation level.
This paper uses this ``mechanism-replacement semantics'' to formalize
defense operations, without conducting causal effect identification based
on do-calculus.
However, in the context of side-channel leakage decomposition, existing
work has not yet unified the ``mechanism replacement'' semantics of SCM
with the chain decomposition of mutual information under the same
computable quantities.
A key theoretical gap is: how to formally map, within the SCM framework,
``which generation mechanism is replaced by a defense'' to ``the projection
of that replacement onto the leakage chain decomposition''?

Information geometry provides a Riemannian geometric framework for
statistical manifolds~\cite{amari2016information}.
The classical work of Amari~et~al. shows that the Fisher information
matrix gives the metric tensor of the statistical manifold, and its
eigenvalue distribution (spectral properties) characterizes the local
anisotropy of the manifold.
In machine learning, information geometry has been widely used for natural
gradient optimization~\cite{amari1998natural}, EM algorithm theory, and
optimization theory for deep learning~\cite{amari1995information}.
However, in the side-channel analysis context, existing work uses the
Fisher information mostly as a local technical component, and a unified
framework that maps ``the mechanism-replacement target of a defense'' and
``the chain increment structure of leakage'' to the same computable quantity
(mutual information/conditional mutual information/coupling information and
their geometric approximations) is lacking.
This paper aims to fill this structured accounting gap: SCM provides
mechanism-replacement semantics, chain decomposition provides leakage
components, and geometric approximation provides an analytic characterization
of local compressibility.

Synthesizing these three threads, this paper is organized around the
following four verifiable theoretical claims:

\begin{enumerate}[label=(\arabic*)]
\item \emph{Formal accounting of leakage structural decomposability.}
  Existing work treats leakage as a holistic quantity $I(X;Y)$.
  This paper formalizes structural decomposability as a computable
  accounting on the chain identity
  $I(X;Y)=I(X;Y_{\mathrm{size}})+I(X;Y_{\mathrm{dir}}|Y_{\mathrm{size}})+
  I(X;Y_{\mathrm{time}}|Y_{\mathrm{size}},Y_{\mathrm{dir}})$
  and makes explicit that this is independent of statistical independence.
  The ``orthogonality'' claimed here refers only to tangent-space
  orthogonality induced by the Fisher metric under a given differentiable
  parametrization; the leakage components themselves use the structural
  decomposition given by the chain rule, with no requirement of statistical
  independence.

\item \emph{Structural explanation of defense failure.}
  Existing work identifies single-dimension defense failure under
  multi-dimensional attacks~\cite{mathews2023sok,shen2023subverting}
  only from empirical phenomena.
  This paper provides a testable sufficient condition and failure
  evidence using coupling information $C_{\mathrm{dir,time}}$ and Markov
  residual $I(X;Y_{\mathrm{time}}|Y_{\mathrm{dir}})$:
  if the Markov residual is significantly greater than zero, then a
  single-dimension direction defense does not satisfy the sufficient
  condition for severing timing leakage.

\item \emph{Structural integration of theoretical tools.}
  This paper aligns the ``mechanism-replacement semantics'' of causal
  inference with the ``Fisher approximation'' of information geometry to
  the same decomposition quantities (mutual information/conditional mutual
  information/coupling information), providing a unified theoretical
  framework for the side-channel leakage decomposition problem.

\item \emph{Quantifiable guidance for defense design.}
  This paper converts ``how much leakage is suppressed per unit cost in
  dimension $d$'' into a reportable metric $\eta_d$ via causal efficacy,
  providing directional guidance for the trade-off design of
  multi-dimensional joint defenses.
\end{enumerate}

It is worth emphasizing that network side-channel defense faces an
unavoidable ``system constraint triangle'': \textbf{bandwidth overhead}
(padding extra bytes occupies link capacity), \textbf{latency overhead}
(random delay affects user experience), and \textbf{buffer stability}
(traffic shaping that continuously buffers packets causes queue overflow).
A defender cannot arbitrarily ``inject noise''---any mechanism replacement
must be performed within the triangular constraints, which is precisely why
``causal efficacy'' and ``cost functional'' appear in this paper: we need
to maximize the leakage suppression per unit cost within the feasible
constraint set.

This paper does not replace attack experiments and does not claim global
conclusions across memory-bearing or non-causal defense classes.
Within the stationary memoryless defense class
(Definition~\ref{def:memoryless-defense}), we establish a formal theory
of leakage multi-dimensionality, providing causal and geometric dual
theoretical foundations for the design of multi-dimensional joint defenses.

\section{Causal Decoupling Model}\label{sec:scm}

This section establishes a causal model from application semantics to
observed features, formalizing the causal dependencies among dimensions.
The core idea is to decompose the observed feature~$Y$ into three atomic
components---size, direction, and timing---using an SCM to characterize
their causal dependencies; and to formalize, without invoking do-calculus,
the effect of defense operations on generation mechanisms using
``mechanism-replacement'' semantics.
The term ``decoupling'' here refers only to the modular decomposition of
generation mechanisms at the structural-equation level in order to specify
intervention targets, and does not presuppose statistical or conditional
independence among the observed dimensions.

\subsection{Multi-Dimensional Decomposition of Side-Channel Features}

In encrypted-traffic side-channel analysis, an adversary records, within a
session window, the packet-level metadata sequence
$\{(t_i,\ell_i,d_i)\}_{i=1}^{n}$---the arrival time, length, and
direction of the $i$-th packet---which is fully visible at the network
layer even when traffic is encrypted.
The adversary extracts three classes of statistics from this, forming the
multi-dimensional decomposition of the observed feature $Y$.

\begin{definition}[Atomic Dimension Decomposition of Side-Channel Features]
\label{def:multi-dim}
The lowest-level atomic metadata of encrypted traffic is the sequence
$(t_i, \ell_i, d_i)$, representing the arrival time, length, and direction
of the $i$-th packet (where $+1$ denotes downstream/inbound and $-1$
denotes upstream/outbound).
We decompose the observed feature $Y$ into three \emph{atomic dimensions},
each corresponding to a single-component statistic of one class of
underlying metadata:
\begin{equation}\label{eq:multi-dim}
Y = (Y_{\mathrm{size}}, Y_{\mathrm{dir}}, Y_{\mathrm{time}}),
\end{equation}
where:
\begin{itemize}
\item \textbf{Size dimension} $Y_{\mathrm{size}}$: A statistical
  description of the \emph{purely} packet-length sequence
  $\{\ell_i\}_{i=1}^n$, \emph{without} mixing direction information.
  Examples include packet-length distribution $P_{\mathrm{size}}$, mean
  length $\bar\ell$, length entropy $H_{\mathrm{size}}$, or cumulative
  packet length $\sum_{i=1}^n \ell_i$.

\item \textbf{Direction dimension} $Y_{\mathrm{dir}}$: A statistical
  description of the \emph{purely} direction sequence $\{d_i\}_{i=1}^n$.
  Examples include the incoming/outgoing packet count ratio
  $n_{\mathrm{in}}/n_{\mathrm{out}}$, the number of direction reversals
  $\#\{i: d_i \neq d_{i+1}\}$, alternation-pattern entropy
  $H_{\mathrm{dir}}$, or the cumulative signed-packet-count sum
  $\sum_{i=1}^n d_i$.

\item \textbf{Timing dimension} $Y_{\mathrm{time}}$: A statistical
  description of the inter-arrival-interval sequence
  $\{\Delta_i\}_{i=1}^{n-1}$, where $\Delta_i=t_{i+1}-t_i$.
  Examples include interval distribution $P_{\mathrm{time}}$, mean
  interval $\bar\Delta$, timing entropy $H_{\mathrm{time}}$, or session
  duration $T=\sum_{i=1}^{n-1}\Delta_i$.
\end{itemize}
This decomposition differs from classical ``signed length'' features
($d_i \times \ell_i$, mixing size and direction) or ``burst statistics''
(mixing direction reversals, duration, and packet counts): the three atomic
dimensions are non-overlapping at the feature-definition level, each
corresponding to one coordinate component of the packet metadata
$(\ell_i, d_i, t_i)$, which enables a precise characterization of
causal dependencies and coupling relationships among dimensions.
\end{definition}

\begin{remark}
As an example (page loading), consider $X=$ ``browser initiating a page
request.''
A text-heavy blog produces a ``short upstream request + single large
downstream'' pattern: $Y_{\mathrm{size}}$ exhibits a bimodal distribution
of few small packets plus many medium ones; $Y_{\mathrm{dir}}$ shows
extreme asymmetry ($n_{\mathrm{in}}/n_{\mathrm{out}}\approx 50$);
$Y_{\mathrm{time}}$ shows dense downstream bursts following a short RTT.
A resource-rich news portal produces multiple request--response alternations:
$Y_{\mathrm{size}}$ has a more complex multimodal distribution;
$Y_{\mathrm{dir}}$ shows frequent direction reversals;
$Y_{\mathrm{time}}$ shows multiple RTT peaks.
Among the three atomic dimensions, $Y_{\mathrm{dir}}$ is usually the most
semantically discriminative, directly exposing the interaction pattern;
$Y_{\mathrm{size}}$ degenerates under fixed-packet-size protocols such as
Tor's 512-byte cells; $Y_{\mathrm{time}}$ captures residual leakage from
network conditions and computation latency.
\end{remark}

\subsection{Structural Causal Model}

The SCM formalizes ``how application semantics determine features in each
dimension.''

\begin{definition}[Side-Channel SCM]\label{def:scm}
The side-channel SCM is a directed acyclic graph (DAG)
$\mathcal{G}=(\mathcal{V},\mathcal{E})$ with vertex set
$\mathcal{V}=\{X, Y_{\mathrm{size}}, Y_{\mathrm{dir}}, Y_{\mathrm{time}}\}$
and edge set $\mathcal{E}$ capturing causal dependencies:
\begin{align}\label{eq:scm-edges}
\mathcal{E} = \{&X\to Y_{\mathrm{size}},\; Y_{\mathrm{size}}\to Y_{\mathrm{dir}},\;
Y_{\mathrm{dir}}\to Y_{\mathrm{time}},\notag\\
&X\to Y_{\mathrm{dir}},\;
X \to Y_{\mathrm{time}},\; Y_{\mathrm{size}}\to Y_{\mathrm{time}}\}.
\end{align}
Each edge corresponds to a structural equation; the defense strategy
variable $D$ is introduced to characterize the effect of different defense
mechanisms:
\begin{align}
Y_{\mathrm{size}} &= f_{\mathrm{size}}(X, U_{\mathrm{size}}; D), \label{eq:f-size}\\
Y_{\mathrm{dir}}  &= f_{\mathrm{dir}}(Y_{\mathrm{size}}, X, U_{\mathrm{dir}}; D), \label{eq:f-dir}\\
Y_{\mathrm{time}} &= f_{\mathrm{time}}(Y_{\mathrm{dir}}, Y_{\mathrm{size}}, X, U_{\mathrm{time}}; D), \label{eq:f-time}
\end{align}
where $U_{\mathrm{size}}, U_{\mathrm{dir}}, U_{\mathrm{time}}$ are
exogenous noise variables (modeling randomness in the protocol stack,
network jitter, Nagle algorithm, etc.) that are statistically independent
of $X$; exogenous noise terms may be correlated with each other to absorb
shared random sources and common network perturbations---the chain MI
decomposition does not depend on the independence of exogenous noise.
$D$ is an externally controllable defense strategy variable (with no
parents in the DAG, selected by the defense designer), whose different
values correspond to different types of defense mechanisms (e.g.,
fixed-length padding, dummy packet injection, random delay).
This paper adopts ``mechanism-replacement'' rather than ``variable
intervention'' semantics: changing $D$ replaces the mechanism parameters
of $f_{\cdot}$ (and possibly the noise distribution), while the functional
form of the structural equations remains unchanged.
\end{definition}

\begin{figure}[!t]
\centering
\resizebox{\columnwidth}{!}{%
\begin{tikzpicture}[
    node/.style={circle, draw=copenblue, thick, minimum size=1.2cm,
                 fill=copenlight!30},
    arrow/.style={-Stealth, very thick, copendark}
]
\node[node] (X)     {\small $X$};
\node[node] (Ysize) [right=2.4cm of X]     {\small $Y_{\mathrm{size}}$};
\node[node] (Ydir)  [right=2.4cm of Ysize] {\small $Y_{\mathrm{dir}}$};
\node[node] (Ytime) [right=2.4cm of Ydir]  {\small $Y_{\mathrm{time}}$};

\draw[arrow] (X) -- (Ysize);
\draw[arrow] (Ysize) -- (Ydir);
\draw[arrow] (Ydir) -- (Ytime);

\draw[arrow, dashed] (X)     to[bend left=18] (Ydir);
\draw[arrow, dashed] (X)     to[bend left=22] (Ytime);
\draw[arrow, dashed] (Ysize) to[bend left=18] (Ytime);

\node[font=\tiny, copendark] (Usize) [above=0.9cm of Ysize] {$U_{\mathrm{size}}$};
\node[font=\tiny, copendark] (Udir)  [above=0.9cm of Ydir]  {$U_{\mathrm{dir}}$};
\node[font=\tiny, copendark] (Utime) [above=0.9cm of Ytime] {$U_{\mathrm{time}}$};

\draw[arrow, dashed, gray, thin] (Usize) -- (Ysize);
\draw[arrow, dashed, gray, thin] (Udir)  -- (Ydir);
\draw[arrow, dashed, gray, thin] (Utime) -- (Ytime);

\node[below=0.15cm of X,     font=\scriptsize] {App. semantics};
\node[below=0.15cm of Ysize, font=\scriptsize] {Size features};
\node[below=0.15cm of Ydir,  font=\scriptsize] {Direction features};
\node[below=0.15cm of Ytime, font=\scriptsize] {Timing features};
\end{tikzpicture}}%
\caption{DAG of the side-channel SCM. Solid arrows represent the
primary causal chain
$X \to Y_{\mathrm{size}} \to Y_{\mathrm{dir}} \to Y_{\mathrm{time}}$;
dashed arrows represent direct causal dependencies; $U$ denotes
exogenous noise.}
\label{fig:scm}
\end{figure}
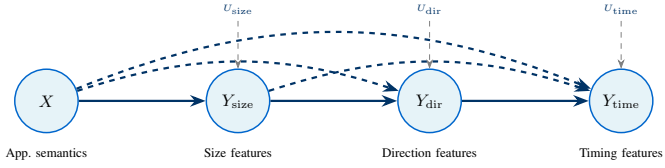

The DAG shown in Fig.~\ref{fig:scm} contains direct edges from $X$ to
each observed variable, a semi-fully-connected structure that deliberately
does not impose strong conditional independence constraints (via
d-separation).
This is an intentional design choice: the complexity of side-channel
leakage implies that $X$ may influence each observed dimension through
multiple paths (direct or indirect).
The SCM functions here primarily as (1) specifying the directionality of
causal dependencies, avoiding confusion of causal direction; and (2)
providing a formal framework for ``mechanism replacement,'' where changing
$D$ corresponds to changing the parameters of specific generation mechanisms
($f_{\mathrm{size}}$, $f_{\mathrm{dir}}$, $f_{\mathrm{time}}$).

\subsection{Defense Strategy Variable and Mechanism Replacement}

\begin{definition}[Defense Strategy Constraints]\label{def:defense-constraint}
Within the SCM framework, different types of defenses correspond to
different constraints on the defense strategy variable $D$:
\begin{enumerate}[label=(\roman*)]
\item \textbf{Size-padding defense}: $D$ only modifies parameters
  associated with the size mechanism (e.g., padding policy, packet-size
  regularization rules) while leaving direction and timing mechanism
  parameters unchanged (e.g., Tamaraw's fixed-length padding, Tor's
  512-byte cells). The mechanism parameters of $f_{\mathrm{size}}$ are
  replaced.

\item \textbf{Direction-obfuscation defense}: $D$ only modifies parameters
  associated with the direction mechanism (e.g., dummy packet injection
  policy, bidirectional padding flow parameters) while leaving size and
  timing mechanism parameters unchanged. The mechanism parameters of
  $f_{\mathrm{dir}}$ are replaced. Typical examples include FRONT's
  dummy packet insertion altering upstream/downstream interaction patterns.

\item \textbf{Timing-delay defense}: $D$ only modifies parameters
  associated with the timing mechanism (e.g., delay distribution, rate
  cap) while leaving size and direction mechanism parameters unchanged
  (e.g., WTF-PAD's random delay, constant-rate sending).
  The mechanism parameters of $f_{\mathrm{time}}$ are replaced.

\item \textbf{Dummy-packet injection defense}: This is a composite
  mechanism-replacement operation. Note that dummy packet injection can
  serve as an implementation of ``direction mechanism replacement''
  (e.g., FRONT primarily targets the interaction/direction structure) or
  can be designed to serve size and timing jointly. In the latter case,
  the injection policy explicitly spans multiple atomic dimensions.
\end{enumerate}
\end{definition}

The advantage of this formulation is that it does not assume ``only the
output of one dimension changes while the generation mechanisms of other
dimensions remain completely unchanged'' (which is false in practice
because the causal chain
$Y_{\mathrm{size}}\to Y_{\mathrm{dir}}\to Y_{\mathrm{time}}$
means that input changes propagate to outputs downstream), but instead
precisely specifies ``which class of mechanism parameters is changed.''

\begin{definition}[Stationary Memoryless Defense]\label{def:memoryless-defense}
Let $\Delta_i := t_{i+1}-t_i$ ($i=1,\dots,n-1$) and $\Delta_n:=0$, so
that $(\ell_i,d_i,\Delta_i)$ is the metadata triple associated with the
$i$-th packet.
A defense mechanism $D$ is called \emph{stationary memoryless} if there
exists an i.i.d.\ noise sequence $\{\xi_i\}_{i=1}^{\infty}$
($\xi_i$ independent of $X$) such that, for any application semantics
$X=x$, the packet-level metadata triple
$Z_i := (\ell_i,d_i,\Delta_i)$ satisfies
\begin{equation}\label{eq:iid-defense}
Z_i \mid (X=x,D) \stackrel{\mathrm{i.i.d.}}{\sim} P_{Z\mid X=x}^{D},
\end{equation}
where $P_{Z\mid X=x}^{D}$ allows $(\ell_i,d_i,\Delta_i)$ to be correlated
at the same index $i$ (no requirement that $P_{\ell,d,\Delta\mid X}^{D}$
factorizes), and the distribution may depend on $D$ to accommodate
randomization mechanisms such as dummy packet injection.
This assumption excludes explicit history dependence (e.g., buffer queue
state) and provides an idealized approximation of a class of analyzable
defense families.
Extending the framework to memory-bearing or non-causal defenses is a
direction for future work.
\end{definition}

\section{Information Geometry Framework}\label{sec:ig}

\textbf{This section answers a purely side-channel question}: when the
defense knob is adjusted by a small step, how does the attacker's
separability change?
The core of the answer is that the KL divergence (the measure of
distributional separability) has a quadratic local approximation governed
by the Fisher information matrix.
Therefore, the size of the Fisher matrix eigenvalues in the defense
parameter direction directly characterizes ``how much separability change
the same magnitude of mechanism replacement can produce''---this is a
local sensitivity analysis tool for side-channel defense design.

\subsection{Statistical Manifold and Fisher Information Matrix}

Given a parametric distribution family
$\{p(y;\theta) : \theta\in\Theta\}$, where
$\theta=(\theta^1,\ldots,\theta^d)$ is the parameter vector and
$\Theta\subset\R^d$ is the parameter space, this family naturally
constitutes a $d$-dimensional differentiable manifold, called the
\emph{statistical manifold}.

In the side-channel analysis context, the parameter $\theta$ should
correspond to observable traffic statistics.
For example, for a packet-length distribution (e.g., exponential or Gamma),
$\theta$ may take the mean parameter $\mu_{\ell}$ and shape parameter
$k_{\ell}$; for a packet inter-arrival distribution (e.g., a Poisson
process), $\theta$ may take the arrival rate $\lambda$; for defense
mechanisms, $\theta$ may include the padding rate $\alpha$ and timing
perturbation scale $\tau_{\mathrm{time}}$.

\begin{definition}[Fisher Information Matrix]\label{def:fim}
For a distribution $p(y;\theta)$, the Fisher information matrix
$G(\theta)=[g_{ij}(\theta)]$ is defined as
\begin{equation}\label{eq:fim}
g_{ij}(\theta) = \E_{Y\sim P_\theta}\!\left[\frac{\partial\log p(Y;\theta)}{\partial\theta^i}
\frac{\partial\log p(Y;\theta)}{\partial\theta^j}\right].
\end{equation}
\end{definition}

The Fisher information matrix gives the Riemannian metric of the
statistical manifold
$\mathcal{M}=\{p(y;\theta):\theta\in\Theta\}$:
for tangent vectors $v,w\in T_\theta\mathcal{M}$, their inner product is
$\langle v,w\rangle_\theta = v^T G(\theta) w$.
The geometric meaning is that the larger the eigenvalue of $G(\theta)$ in
the direction $\theta$, the greater the statistical distance caused by a
small change in parameter $\theta$.
In the side-channel context, $G(\theta)$ is a \emph{distinguishability
sensitivity meter}: if $v^\top G(\theta)v$ is large in some direction $v$,
then a small parameter shift of $\delta$ along $v$ produces a significant
distributional difference (KL divergence $\sim\frac{1}{2}\delta^2 v^\top G(\theta)v$),
and the adversary needs only a limited number of samples to distinguish;
conversely, directions where $v^\top G(\theta)v$ is small are ``blind
zones'' for the adversary.
Each eigenvalue of $G(\theta)$ corresponds to the signal-to-noise ratio
of a ``side-channel distinguishing channel.''

\subsection{Local Regularity and Geometric Approximation of Mutual Information}

\begin{assumption}[Local Second-Order Regularity]\label{ass:local-euclidean}
The log-likelihood $\ell(y;\theta)=\log p(y;\theta)$ is three-times
differentiable in $\theta$ within a neighborhood of $\theta$, with bounded
third-order derivatives, and the Fisher information matrix $G(\theta)$ is
continuous in that neighborhood.
Under this condition, the KL divergence has a second-order Taylor expansion
in the parameter perturbation $\delta$, with the quadratic term given by
the Fisher information.
\end{assumption}

\begin{proposition}[KL Divergence and Fisher Information]\label{prop:kl-fim}
Under Assumption~\ref{ass:local-euclidean}, for a small perturbation
$\|\delta\|\leq \varepsilon$, the second-order expansion of the KL
divergence satisfies:
\begin{equation}\label{eq:kl-fim}
D_{\mathrm{KL}}\!\big(P_{\theta}\,\|\,P_{\theta+\delta}\big)
  = \frac{1}{2}\delta^\top G(\theta)\delta + O(\|\delta\|^3).
\end{equation}
\end{proposition}

\begin{proof}
From the KL divergence definition, with $\ell(y;\theta)=\log p(y;\theta)$:
\[
D_{\mathrm{KL}}\!\big(P_{\theta}\,\|\,P_{\theta+\delta}\big)
  =\E_{Y\sim P_\theta}\!\left[\ell(Y;\theta)-\ell(Y;\theta+\delta)\right].
\]
Applying a second-order Taylor expansion of $\ell(Y;\theta+\delta)$ at $\theta$:
\[
\ell(Y;\theta+\delta)=\ell(Y;\theta)+\delta^\top \nabla_\theta \ell(Y;\theta)
  +\tfrac{1}{2}\delta^\top \nabla_\theta^2\ell(Y;\theta)\delta+O(\|\delta\|^3).
\]
Taking expectations under $P_\theta$ and using the score zero-mean property
$\E_{P_\theta}[\nabla_\theta \ell(Y;\theta)]=0$ and the Fisher information
identity $G(\theta)=-\E_{P_\theta}[\nabla_\theta^2\ell(Y;\theta)]$:
\[
D_{\mathrm{KL}}\!\big(P_{\theta}\,\|\,P_{\theta+\delta}\big)
  =\tfrac{1}{2}\delta^\top G(\theta)\delta+O(\|\delta\|^3). \qed
\]
\end{proof}

Applying the above local expansion to the conditional distribution
$p(y;\theta_x)$ in a neighborhood of the mixture-distribution parameter
$\bar\theta$, and assuming the conditional distributions belong to the same
regular parametric family with $\sup_x \|\theta_x-\bar\theta\| \leq \varepsilon$
and $G(\bar\theta)$ approximating $G(\theta_x)$ throughout the neighborhood
(error absorbed into $O(\varepsilon^3)$):
\begin{equation}\label{eq:mi-approx}
I(X;Y) = \E_x\!\left[\tfrac{1}{2}(\theta_x-\bar\theta)^\top
  G(\bar\theta)(\theta_x-\bar\theta)\right] + O\!\left(\varepsilon^3\right).
\end{equation}

\begin{remark}
The trace form below is used to characterize the local sensitivity
direction of defense parameters, not as a primary empirical tool.
The mutual information quantities in the empirical section are estimated
directly; the trace approximation is used only for interpretive comparison.
\end{remark}

\begin{proposition}[Covariance-Form MI Approximation]\label{prop:mi-cov}
Define the parameter covariance matrix
$\Sigma_\theta := \E_{x\sim P_X}[(\theta_x-\bar\theta)(\theta_x-\bar\theta)^\top]$.
Under the small-perturbation assumptions above:
\begin{equation}\label{eq:mi-trace}
I(X;Y) = \frac{1}{2\ln 2}\mathrm{Tr}\!\big(G(\bar\theta)\Sigma_\theta\big) + O(\varepsilon^3).
\end{equation}
The factor $\frac{1}{\ln 2}$ converts nats to bits; all mutual
information quantities in this paper are in bits.
\end{proposition}

\begin{proof}
Using the trace property $\mathrm{Tr}(AB) = \mathrm{Tr}(BA)$ and the
interchangeability of trace and expectation:
\begin{align*}
\E_x\!\left[(\theta_x-\bar\theta)^\top G(\bar\theta)(\theta_x-\bar\theta)\right]
  &= \mathrm{Tr}\!\big(G(\bar\theta)\E_x[(\theta_x-\bar\theta)(\theta_x-\bar\theta)^\top]\big)\\
  &= \mathrm{Tr}\!\big(G(\bar\theta)\Sigma_\theta\big).
\end{align*}
Substituting into~\eqref{eq:mi-approx} and dividing by $\ln 2$ gives the
result. \qed
\end{proof}

Equation~\eqref{eq:mi-trace} has a clear geometric physical meaning:
side-channel leakage depends on the product of two factors:
(1) the Fisher information matrix $G(\bar\theta)$, characterizing the local
curvature/sensitivity of the statistical manifold, reflecting the
distinguishability of different directions in parameter space;
(2) the covariance matrix $\Sigma_\theta$, characterizing the spread of
application semantics $X$ in parameter space, reflecting the distinguishability
of different applications by their traffic features.
Leakage is significant only when the high-sensitivity directions of the
Fisher metric coincide with the principal directions of semantic spread.

\begin{proposition}[Fisher Information Chain Rule]\label{prop:fisher-chain}
For observed variables $Y=(Y_{\mathrm{size}},Y_{\mathrm{dir}})$, if the
joint density factorizes as
$p(y_{\mathrm{size}},y_{\mathrm{dir}};\theta)=p(y_{\mathrm{size}};\theta)\,
p(y_{\mathrm{dir}}\mid y_{\mathrm{size}};\theta)$
and both the marginal and conditional distributions are regularly
parametrized by the \emph{same} parameter $\theta$ (satisfying the
usual regularity conditions: support independent of $\theta$, score
function exists and integration and differentiation can be exchanged),
then the Fisher information matrix satisfies:
\begin{equation}\label{eq:fisher-chain}
G_{Y}(\theta) = G_{Y_{\mathrm{size}}}(\theta)
  + \E_{Y_{\mathrm{size}}\sim p(\cdot;\theta)}\!\Big[\,G_{Y_{\mathrm{dir}}\mid Y_{\mathrm{size}}}(\theta;Y_{\mathrm{size}})\,\Big],
\end{equation}
where $G_{Y_{\mathrm{dir}}\mid Y_{\mathrm{size}}}(\theta;y_{\mathrm{size}})$
is the conditional Fisher information matrix given
$Y_{\mathrm{size}}=y_{\mathrm{size}}$.
\end{proposition}

\begin{proof}
The Fisher matrix is
$G_Y(\theta) = \E[\nabla_\theta \log p(Y;\theta)
  \nabla_\theta \log p(Y;\theta)^\top]$.
By the chain rule:
$\nabla_\theta \log p(Y;\theta) = \nabla_{\mathrm{size}} + \nabla_{\mathrm{dir}|\mathrm{size}}$,
where $\nabla_{\mathrm{size}} = \nabla_\theta \log p(Y_{\mathrm{size}};\theta)$
and $\nabla_{\mathrm{dir}|\mathrm{size}} = \nabla_\theta \log p(Y_{\mathrm{dir}}\mid Y_{\mathrm{size}};\theta)$.
Expanding and using $\E[\nabla_{\mathrm{dir}|\mathrm{size}}\mid Y_{\mathrm{size}}]=0$
(score zero-mean for the conditional), the cross terms vanish.
The first term is $G_{Y_{\mathrm{size}}}(\theta)$.
The second term equals
$\E_{Y_{\mathrm{size}}}[G_{Y_{\mathrm{dir}}\mid Y_{\mathrm{size}}}(\theta;Y_{\mathrm{size}})]$
by iterated expectation, yielding the result. \qed
\end{proof}

\begin{corollary}[Iterated Generalization of the Fisher Chain Rule]\label{cor:fisher-chain-3}
For $Y=(Y_1,Y_2,Y_3)$ with joint density
$p(y_1,y_2,y_3;\theta)=p(y_1;\theta)p(y_2\mid y_1;\theta)p(y_3\mid y_1,y_2;\theta)$:
\begin{align}
G_{(Y_1,Y_2,Y_3)}(\theta) &= G_{Y_1}(\theta)
  + \E_{Y_1}\!\big[G_{Y_2\mid Y_1}(\theta;Y_1)\big]\notag\\
  &\quad + \E_{Y_1,Y_2}\!\big[G_{Y_3\mid Y_1,Y_2}(\theta;Y_1,Y_2)\big].
\end{align}
\end{corollary}

\begin{proof}
Apply Proposition~\ref{prop:fisher-chain} to $(Y_1,(Y_2,Y_3))$, then again
to $(Y_2,Y_3)$ conditioned on $Y_1$, using iterated expectation. \qed
\end{proof}

\section{Leakage Decomposition and Causal Efficacy}\label{sec:leakage}

This section combines the causal model with information geometry to
establish the leakage decomposition framework and discuss causal efficacy.

\textbf{Notation:} All mutual information and entropy in this paper are in
bits (logarithm base $\log_2$); the Fisher information matrix is defined
using the natural logarithm (the standard convention independent of
parametrization).

We treat $Y_{\mathrm{size}}, Y_{\mathrm{dir}}, Y_{\mathrm{time}}$ as
random variable components extracted from the raw metadata sequence
$Z=\{(t_i,\ell_i,d_i)\}$ via measurable feature maps; they need not be
mutually independent, but each dimension corresponds to a statistical
description of one class of underlying metadata.
By the data-processing inequality, since $X\to Z\to Y$ forms a Markov
chain, feature extraction cannot increase information about $X$, i.e.,
$I(X;Y)\leq I(X;Z)$, so analyzing $I(X;Y)$ is a conservative
representation of the information available to an adversary from raw
metadata.
The chain decomposition depends only on the joint-distribution identity of
random variables.

\subsection{Chain Decomposition of Mutual Information}

\begin{proposition}[Leakage Chain Decomposition]\label{prop:chain-rule}
For any joint distribution $P_{X,Y}$:
\begin{align}\label{eq:chain-rule}
I(X;Y) &= I(X;Y_{\mathrm{size}}) + I(X;Y_{\mathrm{dir}}\mid Y_{\mathrm{size}})\notag\\
        &\quad + I(X;Y_{\mathrm{time}} \mid Y_{\mathrm{size}},Y_{\mathrm{dir}}).
\end{align}
\end{proposition}

\begin{proof}
Apply the MI chain rule $I(A;B,C) = I(A;B) + I(A;C\mid B)$ with
$A=X$, $B=Y_{\mathrm{size}}$, $C=(Y_{\mathrm{dir}}, Y_{\mathrm{time}})$:
\[
I(X;Y) = I(X;Y_{\mathrm{size}}) + I(X;Y_{\mathrm{dir}},Y_{\mathrm{time}}\mid Y_{\mathrm{size}}).
\]
Applying the chain rule again to the second term with $B=Y_{\mathrm{dir}}$,
$C=Y_{\mathrm{time}}$ (conditioned on $Y_{\mathrm{size}}$):
\begin{align*}
I(X;Y_{\mathrm{dir}},Y_{\mathrm{time}}\mid Y_{\mathrm{size}})
  &= I(X;Y_{\mathrm{dir}}\mid Y_{\mathrm{size}})\\
  &\quad + I(X;Y_{\mathrm{time}}\mid Y_{\mathrm{size}},Y_{\mathrm{dir}}).
\end{align*}
Combining gives~\eqref{eq:chain-rule}, which holds for any joint
distribution (see, e.g.,~\cite{cover2006elements}). \qed
\end{proof}

\textbf{Adversary's upgrade path:} The chain decomposition is not merely a
mathematical identity; it is an actionable ``incremental upgrade path'' for
the adversary.
The decomposition holds for any ordering of $Y$ (any permutation yields a
valid chain identity); this paper adopts the ordering
$Y_{\mathrm{size}}\prec Y_{\mathrm{dir}}\prec Y_{\mathrm{time}}$
aligned with the primary causal chain of the SCM for physical
interpretability.
Using only $Y_{\mathrm{size}}$ corresponds to a first-generation fingerprint
attack; adding $Y_{\mathrm{dir}}$ gives a second-generation attack with
incremental leakage $I(X;Y_{\mathrm{dir}}\mid Y_{\mathrm{size}})$;
adding $Y_{\mathrm{time}}$ gives a third-generation attack with incremental
leakage $I(X;Y_{\mathrm{time}}\mid Y_{\mathrm{size}},Y_{\mathrm{dir}})$.
The chain rule guarantees each increment is non-negative, and the three
terms sum exactly to the joint leakage $I(X;Y)$.
The defender's task is to find the ``truncation point'' with the lowest cost
among these incremental steps.

\subsection{Coupling Information and Markov Residual}\label{sec:coupling}

The chain rule decomposes total leakage into three conditional terms.
To discuss the ``compressibility'' of leakage in each dimension, we
analyze the dependences among these terms.

\begin{definition}[Coupling Information]\label{def:coupling}
The coupling information between size and direction is:
\[
C_{\mathrm{size,dir}} := I(Y_{\mathrm{size}};Y_{\mathrm{dir}}|X),
\]
the mutual information between $Y_{\mathrm{size}}$ and $Y_{\mathrm{dir}}$
given the application semantics~$X$.
If $Y_{\mathrm{size}}$ and $Y_{\mathrm{dir}}$ are statistically
independent given $X$, then $C_{\mathrm{size,dir}} = 0$; otherwise
$C_{\mathrm{size,dir}} > 0$ indicates statistical dependence.
Similarly, define $C_{\mathrm{time,size,dir}} = I(Y_{\mathrm{time}};(Y_{\mathrm{size}},Y_{\mathrm{dir}})|X)$.
\end{definition}

Coupling information $C_{\mathrm{size,dir}}$ is not a sufficient condition
for concluding that single-dimension defenses must fail, but it signals
that different dimensions have statistical adhesion at the observation level,
so that single-dimension mechanism replacement may fail to simultaneously
suppress all conditional increment terms.
More precisely: if a defense replaces only the mechanism of $f_{\mathrm{size}}$
(e.g., fixed-length padding) while leaving the interaction pattern of
$f_{\mathrm{dir}}$ unchanged (or only weakly perturbed), then the
direction feature $Y^{(D)}_{\mathrm{dir}}$ generated after the defense is
still produced by the interaction logic correlated with $X$, and
$I(X;Y^{(D)}_{\mathrm{dir}}\mid Y^{(D)}_{\mathrm{size}})$ in general
cannot be driven to zero.
Conversely, when $C_{\mathrm{size,dir}}=0$ (i.e.,
$Y_{\mathrm{size}}\perp Y_{\mathrm{dir}}\mid X$), replacing $f_{\mathrm{size}}$
typically does not significantly change the order of magnitude of the
direction residual term.
Therefore, the larger the coupling information, the more the defense needs
to focus on cross-dimensional joint mechanism replacement---coupling
information is a demand signal for measuring ``whether joint defense is
needed.''

\textbf{Testable sufficient condition for a single-dimension size defense to
suppress direction residual leakage:}
For strategy $\delta$, a sufficient condition for a single-dimension size
defense to simultaneously drive direction residual leakage to zero is that
the post-defense distribution $P^{(\delta)}$ satisfies the Markov condition
\[
X \to Y^{(\delta)}_{\mathrm{size}} \to Y^{(\delta)}_{\mathrm{dir}},
\]
i.e., $I_{P^{(\delta)}}(X;Y^{(\delta)}_{\mathrm{dir}}\mid Y^{(\delta)}_{\mathrm{size}})=0$.
This condition is testable: in the empirical section, we use plug-in MI
estimation with bootstrap confidence intervals to test whether this quantity
is significantly greater than zero.

\subsection{Gaussian Illustrative Example}\label{sec:toy}

To assign concrete side-channel meanings to all abstract symbols, we
present a closed-form Gaussian illustrative example.
Let application semantics $X\in\{A,B\}$ (e.g., ``blog page'' vs.\ ``news
portal''), prior $P_X=\mathrm{Uniform}$.

\textbf{Size component:}
\[
Y_{\mathrm{size}}\mid X=A \sim \mathcal{N}(\mu_A,\sigma^2), \quad
Y_{\mathrm{size}}\mid X=B \sim \mathcal{N}(\mu_B,\sigma^2),
\]
with $\mu_A \neq \mu_B$ (e.g., $\mu_A=500$~bytes, $\mu_B=1400$~bytes),
and $\sigma^2$ representing randomness from the protocol stack.
Parameter $\theta_{\mathrm{size}}=\mu$, Fisher information $G_{\mathrm{size}}=1/\sigma^2$,
semantic spread $\Sigma_{\theta,\mathrm{size}} = (\mu_A-\mu_B)^2/4$.
Under the small-separation setting ($|\mu_A-\mu_B|\ll\sigma$):
\[
I(X;Y_{\mathrm{size}})\approx \frac{(\mu_A-\mu_B)^2}{8\sigma^2\ln 2} \text{ (bits)}.
\]
\textbf{Tor degeneration:} When Tor's fixed 512-byte cells render
$\mu_A=\mu_B=512$ and length variance is negligible,
$I(X;Y_{\mathrm{size}})\approx 0$---the size dimension degenerates strongly
under cell-level observation, which is the key background for the Wang
dataset.

\textbf{Direction component (conditioned on size):}
\[
Y_{\mathrm{dir}}\mid (X,Y_{\mathrm{size}}) \sim \mathcal{N}(\alpha Y_{\mathrm{size}} + \beta_X,\, \tau^2),
\]
where $\alpha$ represents the size--direction causal coupling (larger
packets induce more downstream data), $\beta_X$ is the application-specific
up/downstream asymmetry offset, and $\tau^2$ is the interaction-pattern
randomness.
Conditional Fisher information $G_{\mathrm{dir}|\mathrm{size}}=1/\tau^2$;
conditional leakage $I(X;Y_{\mathrm{dir}}\mid Y_{\mathrm{size}})\approx
(\beta_A-\beta_B)^2/(8\tau^2\ln 2)$.

\textbf{Explicit computation of coupling information:}
\[
C_{\mathrm{size,dir}} = I(Y_{\mathrm{size}};Y_{\mathrm{dir}}\mid X)
  = \frac{1}{2}\log_2\!\left(1+\frac{\alpha^2\sigma^2}{\tau^2}\right) > 0 \text{ (bits)},
\]
provided $\alpha\neq 0$.
In the \textbf{Tor scenario}: as $\sigma^2\to 0$ (fixed cell size),
$C_{\mathrm{size,dir}}\to 0$---size and direction decouple approximately in Tor,
explaining why the dominant leakage in the Wang dataset comes from the
direction dimension.
This illustrative example concretely illustrates the two defense limit mechanisms:
\emph{template/strict padding} (compressing inter-class spread,
$|\mu_A-\mu_B|\to 0$, making $I(X;Y_{\mathrm{size}})\to 0$ while keeping
$C_{\mathrm{size,dir}}$ unchanged) vs.\ \emph{randomization/noise injection}
(deflating Fisher sensitivity, $\tau^2\nearrow$, making
$C_{\mathrm{size,dir}}\to 0$ while keeping $I(X;Y_{\mathrm{size}})$
unchanged), corresponding to the two geometric components of the
``dual compression effect'' discussed below.

\subsection{Geometric Characterization of Causal Efficacy}

For a dimension $Y_d$ (e.g., $Y_{\mathrm{dir}}$), we ask: can the leakage
in this dimension be effectively compressed by changing the defense
mechanism?

\textbf{Defense mechanism to geometric parameters:}
Partition the parameter space as
$\theta = [\theta_{\mathrm{size}}^\top, \theta_{\mathrm{dir}}^\top, \theta_{\mathrm{time}}^\top]^\top$,
with corresponding block decompositions of $G(\theta)$ and $\Sigma_\theta$.
Let the undefended baseline be $D_0$ with working point $\bar\theta^{(0)}$.
A defense mechanism $D$ induces a new working point
$\bar\theta^{(D)} = \psi(D)$ in the parameter space.

\textbf{Dual compression effect of dummy packet injection:}
On the one hand, the random injection of dummy packets increases the
variance of the observed data, lowering the eigenvalues of the Fisher
information matrix $G(\bar\theta^{(D)})$ (making the model less sensitive
to parameter changes); on the other hand, dummy packets are often used to
normalize features of different semantics $X$ toward a unified template,
which geometrically corresponds to compressing the spread of the parameter
covariance matrix $\Sigma_\theta^{(D)}$ along critical directions.
Both geometric effects jointly drive down the leakage
$\mathrm{Tr}(G\Sigma_\theta)$.

By Proposition~\ref{prop:mi-cov}, the effect of defense strategy $D$ on
leakage can be expressed as:
\begin{equation}\label{eq:defense-mi}
I(X;Y^{(D)}) \approx \frac{1}{2\ln 2}\mathrm{Tr}\!\big(G(\bar\theta^{(D)})\Sigma_\theta^{(D)}\big).
\end{equation}

\begin{definition}[Defense Cost Functional]\label{def:cost}
For defense strategy $D$, define the bandwidth-overhead cost functional:
\[
C(D) := \E_{X}\left[\frac{\sum_{i=1}^{n} \ell_i^{(D)}}{\sum_{i=1}^{n} \ell_i^{(D_0)}} - 1\right],
\]
measuring the expected relative bandwidth overhead.
\end{definition}

\begin{definition}[Causal Efficacy]\label{def:causal-efficacy}
Let the ordering be
$Y_{\mathrm{size}}\prec Y_{\mathrm{dir}}\prec Y_{\mathrm{time}}$,
with $S_{\mathrm{size}}=\varnothing$,
$S_{\mathrm{dir}}=\{Y_{\mathrm{size}}\}$,
$S_{\mathrm{time}}=\{Y_{\mathrm{size}},Y_{\mathrm{dir}}\}$.
Let $\delta\in\mathcal{D}$ be a chosen strategy value, $\delta_0$ the
undefended baseline, and $P^{(\delta)}$ the induced joint distribution.
Define the chain leakage term for dimension $d$ under strategy $\delta$:
\[
I_d(\delta):=I_{P^{(\delta)}}\!\big(X;Y_d \mid S_d\big),
\]
and the causal efficacy:
\[
\eta_d(\delta):=\frac{I_d(\delta_0)-I_d(\delta)}{C(\delta)-C(\delta_0)}.
\]
\end{definition}

\begin{proposition}[Leakage Decomposition--Fisher Geometric Alignment]\label{prop:leakage-isomorphism}
Under the local quadratic approximation, the chain decomposition of mutual
information and the chain decomposition of Fisher information are aligned
term by term:
\begin{multline}
\frac{1}{2\ln 2}\mathrm{Tr}(G_Y\Sigma_\theta)
  = \frac{1}{2\ln 2}\mathrm{Tr}(G_{\mathrm{size}}\Sigma_\theta)\\
  + \frac{1}{2\ln 2}\mathrm{Tr}(G_{\mathrm{dir}|\mathrm{size}}\Sigma_\theta)
  + \frac{1}{2\ln 2}\mathrm{Tr}(G_{\mathrm{time}|(\mathrm{size},\mathrm{dir})}\Sigma_\theta),
\end{multline}
where $G_Y = G_{\mathrm{size}} + G_{\mathrm{dir}|\mathrm{size}} + G_{\mathrm{time}|(\mathrm{size},\mathrm{dir})}$
follows from Corollary~\ref{cor:fisher-chain-3}.
\end{proposition}

\begin{proof}
From Corollary~\ref{cor:fisher-chain-3} with $Y_1=Y_{\mathrm{size}}$,
$Y_2=Y_{\mathrm{dir}}$, $Y_3=Y_{\mathrm{time}}$:
\[
G_Y(\bar\theta) = G_{\mathrm{size}}(\bar\theta)
  + G_{\mathrm{dir}|\mathrm{size}}(\bar\theta)
  + G_{\mathrm{time}|(\mathrm{size},\mathrm{dir})}(\bar\theta).
\]
Post-multiplying both sides by $\Sigma_\theta$ and taking the trace (trace
is linear), then dividing by $\ln 2$ gives the result. \qed
\end{proof}

The physical meaning of this alignment: each term in the chain
decomposition---marginal leakage $I(X;Y_{\mathrm{size}})$, conditional
leakage $I(X;Y_{\mathrm{dir}}\mid Y_{\mathrm{size}})$, and
$I(X;Y_{\mathrm{time}}\mid Y_{\mathrm{size}},Y_{\mathrm{dir}})$---corresponds
geometrically to the product of the corresponding Fisher chain component
with the global parameter covariance matrix.
The conditional Fisher information $G_{\mathrm{dir}|\mathrm{size}}$
naturally characterizes ``the additional information contribution of
$Y_{\mathrm{dir}}$ to $\theta$ given $Y_{\mathrm{size}}$ is already
known,'' directly corresponding to the physical meaning of conditional
mutual information.

\begin{lemma}[Block-Diagonal Fisher Matrix]\label{lem:block-diagonal}
If the Fisher information matrix $G(\theta)$ has a block-diagonal structure
(i.e., $g_{ij}(\theta)=0$ for $i\neq j$ corresponding to different
dimensions), then the parameter subspaces of each dimension are mutually
orthogonal on the statistical manifold.
\end{lemma}

\begin{proof}
For tangent vectors $v, w \in T_\theta\mathcal{M}$ from different
dimensional parameter subspaces:
$\langle v, w \rangle_\theta = v^\top G(\theta) w = 0$
when $G(\theta)$ is block-diagonal (off-diagonal blocks are zero). \qed
\end{proof}

\begin{remark}
(1)~Block-diagonal Fisher information does not directly imply statistical
independence $C_{\mathrm{size,dir}}=0$; the former is a geometric property
of parameter estimation while the latter is a statistical dependence
relationship.
(2)~In practice, size, direction, and timing patterns are strongly
correlated, so the block-diagonal assumption typically does not hold (the
Tor scenario being an exception: fixed cells cause size degeneration with
$C_{\mathrm{size,dir}}\approx 0$).
When the block-diagonal assumption fails, cross terms $g_{ij}\neq 0$ exist,
and the parameter estimation changes in different dimensional directions
interact.
One can find locally Fisher-orthogonal coordinates via eigendecomposition
$G(\theta) = Q\Lambda Q^\top$; in the coordinate system defined by $Q$'s
eigenvectors, the Fisher metric becomes diagonal.
\end{remark}

\section{Empirical Validation}\label{sec:empirical}

This section empirically tests three core theoretical claims:
(1)~The three chain components in the ordering
$\hat{Y}_{\mathrm{dir},1}\prec \hat{Y}_{\mathrm{time}}\prec \hat{Y}_{\mathrm{dir},2}$
satisfy stable magnitude relationships
($\hat{Y}_{\mathrm{dir},1}$ term $>$ $\hat{Y}_{\mathrm{dir},2}$
residual term $>$ Markov residual) on both undefended and defended data,
and the suppression effects of each defense on the three-dimensional
components are highly consistent with their design intents.
(2)~The closed-form formulas of the Gaussian illustrative example
(Section~\ref{sec:toy}) agree with Monte Carlo estimates.
(3)~The FRONT defense does not satisfy the Markov sufficient condition
(Section~\ref{sec:coupling}), and causal efficacy
(Definition~\ref{def:causal-efficacy}) is quantifiable for injection-type
defenses.

\textit{Note on chain ordering:} The theoretical framework
(Section~\ref{sec:leakage}) uses the generation-level ordering
$Y_{\mathrm{size}}\prec Y_{\mathrm{dir}}\prec Y_{\mathrm{time}}$.
The empirical section uses
$\hat{Y}_{\mathrm{dir},1}\prec \hat{Y}_{\mathrm{time}}\prec \hat{Y}_{\mathrm{dir},2}$
for two reasons: first, $Y_{\mathrm{size}}$ degenerates to a constant in
the Tor/fixed-cell scenario ($I(X;Y_{\mathrm{size}})\approx 0$), so
placing it at the chain head contributes no meaningful information and is
omitted; second, the direction proxy is split into a coarse-grained
$\hat{Y}_{\mathrm{dir},1}$ (packet count ratio) and a fine-grained
$\hat{Y}_{\mathrm{dir},2}$ (reversal rate), placed at the chain head and
tail respectively to capture direction information at different granularities.
This is a \emph{proxy-level accounting ordering} serving the operability
demonstration of the framework, without claiming to correspond to the true
causal generation order of the SCM.

\subsection{Experimental Setup and Estimation Methods}

\paragraph{Datasets}
All three datasets are from the website fingerprinting dataset collection
released by Deng~et~al.~\cite{deng2024earlystage}, based on the 95-website
Tor traffic originally collected by Sirinam~et~al.~\cite{sirinam2018deep}.

\begin{itemize}
\item \textbf{CW (Undefended baseline)}~\cite{sirinam2018deep}:
  Standard closed-world dataset, 95 websites, 1051--1125 traces per class
  (full set); compiled and distributed by Deng~et~al.~\cite{deng2024earlystage}.
  Mean flow: $2038\pm 2020$ packets/flow, mean duration $27.9\pm27.2$~seconds.

\item \textbf{FRONT}~\cite{gong2020zero}:
  A zero-delay lightweight defense (Gong \& Wang, USENIX Security~2020).
  Defense mechanism: dummy packets injected at the start of flows; per-direction
  dummy counts sampled from Uniform, with each dummy packet's timing sampled
  from a Rayleigh distribution (window parameter $w$).
  Same 95 websites, 1000 traces per class.
  Measured packet count overhead: $2873\pm 2066$ packets/flow
  (+41\% vs.\ CW); mean duration $27.8\pm 27.0$~seconds (confirming
  zero-delay property).

\item \textbf{TrafficSliver}~\cite{delacadena2020trafficsliver}:
  A multi-path splitting defense (De la Cadena~et~al., CCS~2020;
  TrafficSliver-Net variant).
  Defense mechanism: Tor traffic distributed across multiple guard nodes
  using batch weighted random (BWR) splitting; the adversary can only
  eavesdrop on one guard node.
  Same 95 websites, 1000 traces per class.
  Measured per-path packets: $1071\pm 1447$ packets/flow ($\approx$52.6\%
  of CW, consistent with partial-flow observation under multi-path splitting);
  mean duration $26.8\pm 26.9$~seconds (no additional delay).
\end{itemize}

We subsample 200 traces per class from each dataset (19,000 traces total)
for MI estimation.

\paragraph{Proxy Features}
The Wang dataset contains only packet directions and timestamps, no byte
sizes.
More critically, since it comes from the Tor network---which uses fixed
512-byte cells for all application data---the size dimension
$Y_{\mathrm{size}}$ is \emph{completely degenerate} in physical terms
(all packets are 512 bytes, $I(X;Y_{\mathrm{size}})\approx 0$).
Rather than reproducing the optimal attack feature set, we construct
\textbf{low-dimensional scalar proxy statistics} to demonstrate the
computability of the decomposition and coupling criteria on real data:
\begin{align*}
\hat{Y}_{\mathrm{dir},1} &= \frac{n_{\mathrm{in}}}{n_{\mathrm{total}}},\quad
\hat{Y}_{\mathrm{time}} = \frac{t_{\mathrm{end}}}{n_{\mathrm{total}}-1},\\
\hat{Y}_{\mathrm{dir},2} &= \frac{\#\text{direction-reversal pairs}}{n_{\mathrm{pairs}}}.
\end{align*}
$\hat{Y}_{\mathrm{dir},1}$ captures up/downstream asymmetry (fraction of
incoming packets), proxying $Y_{\mathrm{dir}}$;
$\hat{Y}_{\mathrm{time}}$ is the mean inter-packet interval, proxying
$Y_{\mathrm{time}}$;
$\hat{Y}_{\mathrm{dir},2}$ is the direction reversal rate (half-duplex
switching intensity), another proxy for $Y_{\mathrm{dir}}$.
Note: the two direction proxies ($\hat{Y}_{\mathrm{dir},1}$ and
$\hat{Y}_{\mathrm{dir},2}$) capture different aspects of the direction
sequence (overall asymmetry vs.\ local switching patterns); they are
treated as a joint proxy for $Y_{\mathrm{dir}}$.
All three are weak approximations---not equivalents---of the atomic
components in Definition~\ref{def:multi-dim}.

The theoretical framework always refers to the physical observation
variables $Y=(Y_{\mathrm{size}},Y_{\mathrm{dir}},Y_{\mathrm{time}})$.
The empirical section, constrained by the observability of Wang/Tor data,
constructs $\hat{Y}$ as a computable proxy for $Y$; the chain accounting
of $\hat{Y}$ is intended to demonstrate the computational closure of the
structured decomposition, not to claim equivalence to a full-information
decomposition of $Y$.

\paragraph{MI Estimation and Uncertainty}
\textbf{Estimator:} Discrete plug-in (histogram) MI estimator.
The plug-in estimator is chosen over continuous estimators (e.g., KSG)
for two reasons:
(i)~it guarantees that the chain identity
$I(X;Y_2\mid Y_1)=I(X;Y_1,Y_2)-I(X;Y_1)$
is strictly closed at the estimation level, avoiding systematic bias
accumulation across different estimators;
(ii)~the discrete joint distribution makes bootstrap resampling and
confidence interval construction controllable and interpretable.

\textbf{Binning:} Uniformly 5~bins per dimension ($5^3=125$ joint symbols,
expected sample count $\approx 152$, controlling sparse-symbol bias).
\textbf{Bin boundaries:} Quantile-equal-frequency binning for each continuous
proxy variable to improve robustness against outliers and avoid sparse symbols.
Bin boundaries are estimated from the CW subsample (200$\times$95 traces,
random seed~42) and fixed; then applied to FRONT and TrafficSliver to ensure
cross-mechanism comparability.
Bin boundaries depend only on the marginal distribution quantiles of proxy
variables, not on class labels $X$, so they introduce no label information.
Conditional MI is computed by the difference identity:
first computing $I(X;Y_1)$, $I(X;Y_1,Y_2)$, $I(X;Y_1,Y_2,Y_3)$,
then deriving conditional terms.
If a difference is negative (due to finite-sample bias), zero truncation
is applied; the negative-value occurrence rate is $<1\%$, and zero truncation
does not affect the conclusions.

\textbf{Bootstrap:} 200 rounds of flow-level resampling with replacement
for core quantities (Markov residual, coupling information, chain components);
each round recomputes all marginal/joint MI from scratch to reflect the
uncertainty of the entire computation chain.
95\% CIs use the percentile method (2.5\%--97.5\%).
Experiments show that increasing to 500 rounds changes CI width by
$<3\%$, not affecting significance conclusions.

\textbf{Binning sensitivity:} In the range of 5 to 20 bins, the spectral
order of the three-dimensional chain components
($Y_{\mathrm{dir},1}$ dominates $>$ $Y_{\mathrm{dir},2}$ $>$
$I(X;Y_{\mathrm{time}}\mid Y_{\mathrm{dir},1})$)
remains invariant on the CW dataset; numerical values increase
monotonically with bin count (finer bins $\to$ higher estimates), but the
relative ordering and cross-defense comparison conclusions are stable.
As a supplementary check, the KSG estimator ($k=5$ neighbors) for
$I(X;\hat{Y}_{\mathrm{dir},1})$ yields ranking trends consistent with the
plug-in conclusions.

\textbf{Reproducibility:} All experiments use the following fixed settings:
subsampling 200 traces per class, random seed 42; uniform class prior
(equal sampling per class); 5 bins/dimension (125 joint symbols);
flow-level bootstrap resampling 200 rounds, percentile 95\% CI.

\subsection{Experiment~I: Gaussian Illustrative Example Validation}

Using the parameters of Section~\ref{sec:toy}
($\mu_A=500$, $\mu_B=1000$, $\sigma=400$, $\alpha=0.3$, $\tau=80$,
$N=8000$ per class; separation ratio $|\mu_A-\mu_B|/\sigma=1.25$),
Table~\ref{tab:toy} compares the theoretical formulas against Monte Carlo
estimates (in bits, $\log_2$).

\begin{table}[!t]
\centering\small
\caption{Gaussian Illustrative Example: Theoretical Formulas vs.\ Monte Carlo Estimates}
\label{tab:toy}
\begin{tabular}{p{4.2cm}ccc}
\toprule
Quantity & Theory & Estimate & Err. \\
\midrule
$I(X;Y_{\mathrm{size}})$ (2nd-order)\newline
$=(\mu_A\!-\!\mu_B)^2/(8\sigma^2\ln 2)$
  & $0.195$ & $0.238$ & $+22\%$ \\[2pt]
$C_{\mathrm{size,dir}} = \tfrac{1}{2}\log_2\!\left(1\!+\!\tfrac{\alpha^2\sigma^2}{\tau^2}\right)$
  & $0.850$ & $0.806$ & $-5\%$ \\
\bottomrule
\end{tabular}
\end{table}

The $+22\%$ positive bias of $I(X;Y_{\mathrm{size}})$ is consistent with
theory: the small-separation approximation
$(\mu_A-\mu_B)^2/(8\sigma^2\ln 2)$ is a lower bound on the true Gaussian
mixture MI (the true value exceeds this approximation at finite
separation, with the error direction being predictable).
The error for $C_{\mathrm{size,dir}}$ is only $-5\%$, validating the
numerical accuracy of the coupling formula.

Two limit mechanisms are also verified:
(i)~\textbf{Inter-class spread collapse} ($\mu_A\to\mu_B$): as
$\mu_A\to\mu_B$, $I(X;Y_{\mathrm{size}})$ decreases monotonically from
$\approx 0.24$~bits to 0, while
$C_{\mathrm{size,dir}}=\frac{1}{2}\log_2(1+\alpha^2\sigma^2/\tau^2)$
is \emph{independent of $\mu_A, \mu_B$} and remains at 0.83--0.84~bits.
This confirms: \emph{compressing inter-class differences eliminates size-dimension
leakage but does not affect coupling} (in the Tor scenario, this effect is
not prominent since $Y_{\mathrm{size}}$ is already degenerate).
(ii)~\textbf{Conditional noise injection} ($\tau\uparrow$): as $\tau$
increases from~80 to~800, $C_{\mathrm{size,dir}}$ decreases monotonically
from~0.82 to~0.04~bits, while $I(X;Y_{\mathrm{size}})$ remains at
$\approx 0.23$~bits throughout.
This confirms: \emph{enhancing direction conditional noise eliminates coupling
but does not affect size-dimension leakage}.

\subsection{Experiments II--III: Real Traffic Chain Decomposition and Causal Efficacy}

\paragraph{Chain Decomposition (Experiment~II)}
Table~\ref{tab:chain} reports the MI chain decomposition results for the
three datasets (5-bin/dimension joint estimation), along with coupling
information and Markov residuals (with bootstrap 95\% CIs; format:
point estimate [CI lower, upper]).
The Markov residual $I(X;\hat{Y}_{\mathrm{time}}\mid\hat{Y}_{\mathrm{dir},1})$
tests whether the condition
$X\to\hat{Y}_{\mathrm{dir},1}\to\hat{Y}_{\mathrm{time}}$ holds at the
proxy-variable level; the augmented test residual
$I(X;\hat{Y}_{\mathrm{time}}\mid \hat{Y}_{\mathrm{dir},1},\hat{Y}_{\mathrm{dir},2})$
is a robustness check that conditions on both direction proxies jointly.

\begin{table}[!t]
\centering\small
\caption{MI Chain Decomposition: Three-Dataset Comparison
(Unified 5-Bin/Dimension Joint Estimation; Unit: bits; Markov Residual
and Coupling Information Rows Include Bootstrap 95\% CI; Augmented
Residual Row Is the Robustness Check)}
\label{tab:chain}
\scriptsize
\begin{tabular}{lrrr}
\toprule
Chain Term & CW & FRONT & TS$^{\star}$ \\
\midrule
$I(X;\hat{Y}_{\mathrm{dir},1})$
  & 0.637 & 0.236 & 0.095 \\
Markov res.\ $I(X;\hat{Y}_{\mathrm{time}}\!\mid\!\hat{Y}_{\mathrm{dir},1})$
  & 0.018$^{a}$ & 0.021$^{b}$ & 0.042$^{c}$ \\
$I(X;\hat{Y}_{\mathrm{dir},2}\!\mid\!\hat{Y}_{\mathrm{dir},1},\hat{Y}_{\mathrm{time}})$
  & 0.219 & 0.200 & 0.167 \\
$I_{\mathrm{total}}$ (chain sum)
  & \textbf{0.874} & \textbf{0.457} & \textbf{0.304} \\
\midrule
$C_{\mathrm{dir,time}}$
  & 0.042$^{d}$ & 0.014$^{e}$ & 0.024$^{f}$ \\
Augm.\ res.\ $I(X;\hat{Y}_{\mathrm{time}}\!\mid\!\hat{Y}_{\mathrm{dir},1},\hat{Y}_{\mathrm{dir},2})$
  & 0.014$^{g}$ & 0.017$^{h}$ & --- \\
\bottomrule
\multicolumn{4}{l}{$^{\star}$TS = TrafficSliver.
  $C_{\mathrm{dir,time}}=I(\hat{Y}_{\mathrm{dir},1};\hat{Y}_{\mathrm{time}}\!\mid\! X)$.}\\
\multicolumn{4}{l}{95\% Bootstrap CI (bits):
  $^{a}$[.015,.022], $^{b}$[.016,.027], $^{c}$[.036,.049],}\\
\multicolumn{4}{l}{$^{d}$[.036,.048], $^{e}$[.010,.018], $^{f}$[.020,.029],
  $^{g}$[.011,.018], $^{h}$[.013,.022].}
\end{tabular}
\end{table}

The spectral order of the three-dimensional chain components across all
three datasets is uniformly:
$I(X;\hat{Y}_{\mathrm{dir},1}) > I(X;\hat{Y}_{\mathrm{dir},2}\mid\cdot)
> I(X;\hat{Y}_{\mathrm{time}}\mid\hat{Y}_{\mathrm{dir},1})$,
and this order is invariant across 5 to 20 bins, indicating that the
\emph{component spectrum} given by the decomposition framework has robustness
to estimation parameters and is not a discretization artifact.
As a robustness check, using 20-bin high-resolution estimation of the
marginal $I(X;\hat{Y}_{\mathrm{dir},1})$, the values for the three datasets
are respectively 1.139, 0.376, 0.233~bits; the relative ordering is fully
consistent with the 5-bin estimate (0.637, 0.236, 0.095), though absolute
values are systematically higher due to finer binning.

\textbf{Mechanism--component alignment:}
FRONT focuses on dummy packet injection to obfuscate the
upstream/downstream interaction pattern;
accordingly, $I(X;\hat{Y}_{\mathrm{dir},1})$ drops substantially
($0.637\to 0.236$, $-62.9\%$), while the Markov residual (chain term~2)
and the fine-grained direction residual (chain term~3) are almost
unchanged ($0.018\to 0.021$, $0.219\to 0.200$), directly mapping to its
SCM interpretation of ``primarily replacing the $f_{\mathrm{dir}}$
mechanism.''
TrafficSliver's multi-path splitting perturbs the direction structure and
timing correlation; the Markov residual (chain term~2) rises significantly
($0.018\to 0.042$; bootstrap 95\% CIs for TrafficSliver $[0.036,0.049]$
and CW $[0.015,0.022]$ do not overlap), while the direction term drops
more ($0.637\to 0.095$, $-85.1\%$), reflecting the strong perturbation
of multi-path splitting on the direction distribution and timing correlation.

\paragraph{Causal Efficacy (Experiment~III)}
\textbf{Cost proxy:} Since the Wang dataset contains no byte lengths,
we use a \textbf{packet-count cost proxy}:
\begin{equation}\label{def:cost-proxy}
C_{\mathrm{pkt}}(D) := \frac{\mathbb{E}[n^{(D)}_{\mathrm{pkt}}]}{\mathbb{E}[n^{(D_0)}_{\mathrm{pkt}}]} - 1.
\end{equation}
The reported $\eta_d$ should be interpreted as ``bits of leakage suppressed
per unit relative packet-count overhead.''

FRONT's measured packet-count overhead is $C_{\mathrm{pkt}}(\mathrm{FRONT})=+41.3\%$.
Table~\ref{tab:eta} reports three-dimensional causal efficacy
($\Delta I_d := I_d(\mathrm{CW}) - I_d(\mathrm{FRONT})$, consistent with
the 5-bin/dimension joint estimation of Table~\ref{tab:chain}).

\begin{table}[!t]
\centering\small
\caption{FRONT Defense: Chain Causal Efficacy
(Packet-Count Overhead $=41.3\%$;
$\eta_d$ Unit: bits per unit relative packet-count overhead;
Unified 5-Bin/Dimension Estimation)}
\label{tab:eta}
\begin{tabular}{p{2.2cm}ccp{2.8cm}}
\toprule
Dim.\ $d$ & $\Delta I_d$ (bits) & $\eta_d$ & Interpretation \\
\midrule
$\hat{Y}_{\mathrm{dir},1}$
  & $+0.401$ & $\mathbf{0.97}$
  & Highly efficient (target dim.) \\[2pt]
$\hat{Y}_{\mathrm{time}}\!\mid\!\hat{Y}_{\mathrm{dir},1}$
  & $-0.003$ & $-0.007$
  & No suppression observed \\[2pt]
$\hat{Y}_{\mathrm{dir},2}\!\mid\!\cdot$
  & $+0.019$ & $0.046$
  & Marginal improvement \\
\bottomrule
\end{tabular}
\end{table}

FRONT achieves $\eta_d\approx 1.0$~bits/(relative packet-count overhead)
on the $\hat{Y}_{\mathrm{dir},1}$ dimension, while the timing and second
direction dimensions have $\eta_d$ near zero or negative---fully
consistent with FRONT's design intent of focusing dummy packet injection
on obfuscating the upstream/downstream interaction pattern without
intervening in the timing scheduling mechanism.
TrafficSliver is a multi-path splitting defense whose main system cost is
establishing multiple circuits (latency and link occupancy) rather than
extra packet injection; its cost functional is therefore incompatible with
the packet-count cost proxy.
We \emph{report comparable $\eta_d$ only for injection-type defenses};
splitting-type defenses should define a cost functional matched to their
mechanism (e.g., total link occupancy, multi-path latency tail) before
computing causal efficacy.

\subsection{Experiment~IV: Coupling Information and Markov Residual}

The last two rows of Table~\ref{tab:chain} report bootstrap CIs for the
coupling information and Markov residual.

\textbf{Test subject:} We test whether the Markov sufficient condition
$X\to \hat{Y}^{(\delta)}_{\mathrm{dir}}\to \hat{Y}^{(\delta)}_{\mathrm{time}}$
from Section~\ref{sec:coupling} holds.
Table~\ref{tab:chain} provides two levels of testing:
chain term~2 $I(X;\hat{Y}_{\mathrm{time}}\mid \hat{Y}_{\mathrm{dir},1})$
(single direction proxy condition) provides the basic test consistent with
the chain accounting;
the last row $I(X;\hat{Y}_{\mathrm{time}}\mid \hat{Y}_{\mathrm{dir},1},\hat{Y}_{\mathrm{dir},2})$
(joint direction proxy condition) provides a more stringent test by
conditioning on more direction information to rule out the alternative
explanation that ``the residual arises from insufficient direction proxy
coverage.''
Both quantities are empirical tests of the Markov chain at the
\emph{proxy-variable level}, not a test of the true high-dimensional
$(Y_{\mathrm{dir}},Y_{\mathrm{time}})$.
``Significantly greater than zero'' means ``this evidential path has not
been severed at the proxy level''; it does not mean ``no direction defense
can ever sever timing leakage.''
Failure of the sufficient condition indicates only that this specific
evidential path is blocked; we do not claim this to be a necessary condition.

\textbf{Key findings:}
After FRONT defense, the Markov residual
$I(X;\hat{Y}_{\mathrm{time}}\mid \hat{Y}_{\mathrm{dir},1}) = 0.021$~bits,
bootstrap 95\% CI $[0.016, 0.027]$ excludes zero;
the more stringent augmented test residual
$I(X;\hat{Y}_{\mathrm{time}}\mid \hat{Y}_{\mathrm{dir},1},\hat{Y}_{\mathrm{dir},2}) = 0.017$~bits,
bootstrap 95\% CI $[0.013, 0.022]$ still excludes zero---indicating the
residual is not caused solely by insufficient direction proxy coverage.
Both levels of testing show that the Markov condition does not hold at the
proxy-variable level.

TrafficSliver's augmented test residual is not reported here: in the
multi-path splitting scenario, each flow is split across different
sub-paths, so $\hat{Y}_{\mathrm{dir},2}$ (the fine-grained direction
distribution proxy) is no longer comparable in meaning to the single-path
CW/FRONT scenarios (the sub-path view's direction statistics reflect
path-assignment logic rather than application semantics); thus, the joint
conditional residual under TrafficSliver lacks an interpretive basis
consistent with CW/FRONT, and cross-scenario comparison is not made.

This result indicates that under FRONT's mechanism family and the
proxy-statistic definitions of this paper, the sufficient condition is not
satisfied empirically.
This is consistent with the theoretical analysis: FRONT replaces only the
direction generation mechanism $f_{\mathrm{dir}}$ (inserting dummy packets
to alter the upstream/downstream interaction pattern), while the timing
scheduling mechanism $f_{\mathrm{time}}$ is not replaced;
the timing proxy $\hat{Y}^{(\delta)}_{\mathrm{time}}$ after defense is
still generated by the scheduling clock correlated with $X$, so conditional
independence $X\perp\!\!\!\perp \hat{Y}^{(\delta)}_{\mathrm{time}}\mid
\hat{Y}^{(\delta)}_{\mathrm{dir}}$ cannot hold automatically, and the
Markov residual cannot be driven to zero.
The coupling information $C_{\mathrm{dir,time}}$ drops from~0.042
to~0.014~bits (95\% CI $[0.010,0.018]$ does not overlap with CW's
$[0.036,0.048]$), confirming that FRONT does weaken the statistical
association between direction and timing features;
but mere reduction of coupling is insufficient to guarantee satisfaction
of the sufficient condition, and the persistently positive Markov residual
shows that the timing residual leakage path remains open.

\subsection{Empirical Summary}

The four experiments systematically validate the three core claims:
\begin{enumerate}[label=(\arabic*)]
\item \textbf{Component spectrum aligned with defense intent}
  (Proposition~\ref{prop:chain-rule} grounded): The direction chain
  component dominates across all three datasets (consistent with the
  expected degeneration of the size dimension in the Tor scenario);
  the spectral order
  $I(X;\hat{Y}_{\mathrm{dir},1}) > I(X;\hat{Y}_{\mathrm{dir},2}\mid\cdot) > I(X;\hat{Y}_{\mathrm{time}}\mid\hat{Y}_{\mathrm{dir},1})$
  is invariant across 5 to 20 bins; the suppression effects of each defense
  on components correspond precisely to their SCM mechanism replacements.

\item \textbf{Gaussian illustrative example accuracy} (Section~\ref{sec:toy}):
  $C_{\mathrm{size,dir}}$ formula error is only~5\%;
  $I(X;Y_{\mathrm{size}})$ formula gives a lower bound in the correct
  direction ($+22\%$, a known property of the small-separation
  approximation); experiments confirm that inter-class spread collapse and
  coupling attenuation are controlled by independent parameters, and
  validate the mechanism by which Tor's fixed cells cause size-dimension
  degeneration and coupling to approach zero.

\item \textbf{Markov sufficient condition not satisfied at proxy level}
  (Section~\ref{sec:coupling}): Under FRONT's mechanism family and the
  proxy-statistic definitions of this paper, the Markov residual is
  0.021~bits with bootstrap 95\% CI $[0.016,0.027]$ excluding zero,
  empirically failing to satisfy the sufficient condition
  $X\to \hat{Y}^{(\delta)}_{\mathrm{dir}}\to \hat{Y}^{(\delta)}_{\mathrm{time}}$,
  consistent with the theory's prediction of the limitations of
  single-dimension (injection-type) direction channel mechanism replacement.
\end{enumerate}

\section{Discussion}\label{sec:discussion}

\subsection{Boundaries of Mechanism Separability}

Causal efficacy (Definition~\ref{def:causal-efficacy}) uses
\emph{mechanism separability} as an additional condition for causal
interpretation (not a prerequisite for computational feasibility):
defense $\delta$ replaces only the generation mechanism parameter family
$f_d$ of dimension $d$, leaving the structural equations $f_{d'}$
($d'\neq d$) of other dimensions unchanged---so that $\eta_d$ can be
interpreted as the counterfactual efficiency of a single mechanism
replacement, rather than the marginal attribution of a multi-dimensional
joint intervention.
The numerical value of $\eta_d$ can still be computed under joint
interventions, but its causal interpretation changes accordingly.

Injection-type defenses satisfy mechanism separability in the
\emph{sense of the structural-equation intervention target}:
FRONT replaces only the direction generation mechanism $f_{\mathrm{dir}}$
and does not modify the true packet transmission clock or scheduling
mechanism $f_{\mathrm{time}}$.
However, this does not imply separability at the \emph{proxy-statistic
level}.
The timing proxy $\hat{Y}_{\mathrm{time}} = t_{\mathrm{end}}/(n_{\mathrm{total}}-1)$
in our empirical analysis explicitly depends on $n_{\mathrm{total}}$,
and FRONT's dummy packet injection changes $n_{\mathrm{total}}$, thereby
perturbing $\hat{Y}_{\mathrm{time}}$ at the proxy-variable level even
though $f_{\mathrm{time}}$ itself is not replaced.
Therefore, the more precise causal explanation for why the Markov residual
remains significantly nonzero after FRONT is: FRONT does not replace
the real packet timing mechanism $f_{\mathrm{time}}$, and the perturbation
it introduces to the proxy timing statistic via changing $n_{\mathrm{total}}$
is insufficient to make the conditional independence
$X\to\hat{Y}_{\mathrm{dir}}\to\hat{Y}_{\mathrm{time}}$ hold at the
proxy level.

When a defense acts simultaneously on multiple dimensions through
\emph{shared network-layer resources}, mechanism separability no longer
holds.
For example, a token-bucket traffic shaper controls the send queue
uniformly, simultaneously determining how packet sizes are truncated and
when packet intervals are released, corresponding to a joint mechanism
replacement of $(f_{\mathrm{size}}, f_{\mathrm{time}})$ rather than an
independent replacement of a single structural equation.
Similarly, TrafficSliver's path-assignment decisions simultaneously affect
timing (propagation delays on different paths) and the observed direction
statistics, another typical coupled intervention.

Under joint interventions, $\eta_d(\delta)$ is still computationally valid
(the numerator $I_d(\delta_0)-I_d(\delta)$ is an empirically measurable
quantity), but its causal interpretation changes fundamentally: it captures
the \emph{marginal observational effect} of the joint intervention on
dimension $d$, rather than the counterfactual leakage compression from
replacing only the single generation mechanism parameter family $f_d$.
For splitting-type defenses, $\eta_d(\delta)$ should be understood as a
descriptive ``effect attribution'' metric rather than mechanistic efficiency
in the causal sense.
A practical criterion for judging separability: if the generation processes
of two dimensions depend on \emph{independent protocol layers} (e.g.,
application logic determines direction, network-layer congestion determines
timing), then mechanism separability holds in first-order approximation;
if they share the same physical scheduling resource, coupling is unavoidable.

\subsection{Valid Range of the Trace Approximation}

The trace approximation of Proposition~\ref{prop:mi-cov},
$I(X;Y)\approx\frac{1}{2\ln 2}\mathrm{Tr}(G(\bar\theta)\Sigma_\theta)$,
has known degeneration modes at both extremes, and understanding these
extremes is critical for correctly using the trace formula.

At the \emph{large-separation end}, the class-conditional parameters
$\theta_x$ are far apart in parameter space, and the truncation error of
the second-order Taylor expansion is comparable to the leading term.
Since $I(X;Y)$ is bounded (by $\log|\mathcal{X}|$) while
$\mathrm{Tr}(G(\bar\theta)\Sigma_\theta)$ grows unboundedly as
$\Sigma_\theta$ increases, the trace approximation systematically
over-estimates leakage in high-separation scenarios---which is precisely
the regime of greatest interest to adversaries where inter-application
differences are large and mutual information is high.

At the \emph{very-small-separation end}, both sides approach zero as
$\Sigma_\theta\to 0$, and the approximation is relatively stable in ratio
terms, but the quantity is too small to be of engineering significance;
more importantly, in this regime the trace formula essentially characterizes
``how leakage would respond if inter-class differences were slightly
increased''---a sensitivity quantity about \emph{direction} rather than
an absolute quantity.

This analysis shows that the appropriate use of the trace formula is to
compare the \emph{directional efficiency} of different defense perturbations
at a \emph{fixed baseline working point $\bar\theta$} (i.e., which defense
perturbation direction most effectively reduces $\mathrm{Tr}(G\Sigma_\theta)$),
rather than comparing the absolute leakage levels of different defense
schemes across baseline points.
In the empirical section of this paper, direct MI estimation is the primary
tool and the trace approximation is used only for conceptual comparison,
precisely to avoid relying on the quantitative accuracy of this approximation
on real data.
In the Tor/Wang scenario, the order of magnitude of the direction--timing
components determines that the trace approximation has reference value only
in the timing dimension (low-separation end), while the direction dimension
(high-separation end) should be estimated directly.

\section{Conclusion}\label{sec:conclusion}

This paper systematically studies the \textbf{structural decomposability}
of encrypted traffic side-channel leakage.
Based on the structural causal model
$X\to Y_{\mathrm{size}}\to Y_{\mathrm{dir}}\to Y_{\mathrm{time}}$
and the mutual-information chain rule
(Proposition~\ref{prop:chain-rule}), total leakage $I(X;Y)$ is precisely
decomposed into three sequential incremental terms for size, direction, and
timing.
Defense operations are formalized via the mechanism-replacement semantics of
defense strategy variable $D$, creating a precise and testable correspondence
between defense design and chain component suppression.
Coupling information $C_{\mathrm{size,dir}}=I(Y_{\mathrm{size}};Y_{\mathrm{dir}}\mid X)$
(Definition~\ref{def:coupling}) is introduced to measure inter-dimensional
statistical dependence; the Markov residual tests the sufficient condition
for a single-dimension defense to sever subsequent leakage.
Causal efficacy $\eta_d$ (Definition~\ref{def:causal-efficacy}) quantifies
the per-unit-cost leakage suppression efficiency.
The Fisher-geometric approximation
$I(X;Y)\approx\frac{1}{2\ln 2}\mathrm{Tr}(G\Sigma_\theta)$
(Proposition~\ref{prop:mi-cov}) is established under a differentiable
parametric family and small-perturbation assumptions.
The Fisher-geometric alignment (Proposition~\ref{prop:leakage-isomorphism})
shows that the chain structures of mutual information and Fisher information
correspond term by term under the second-order trace approximation.
The unification of chain decomposition, coupling measurement, and
Fisher-geometric analysis constitutes a computable, structured leakage
accounting method and a theoretical foundation for multi-dimensional joint
defense design.

The output of this paper is a structured decomposition and testable
criteria, \emph{not} a universal theorem that ``any single-dimension
defense must fail.''
The core structure revealed by the chain decomposition is: a
single-dimension defense replaces the generation mechanism of dimension $d$,
but the conditional terms of the remaining dimensions in the chain leakage
path are not intervened upon, leaving residual paths open.
Coupling information $C_{\mathrm{size,dir}}$ is the leading-indicator
warning of this failure mechanism: the larger its value, the harder it is
for single-dimension mechanism replacement to simultaneously suppress all
conditional increment terms; the Markov residual is its testable
posterior criterion.
The appropriate use of the trace approximation is to compare the
directional efficiency of different defense perturbations at a fixed working
point, not to compare absolute leakage values across baseline points;
causal efficacy $\eta_d$ provides a scalar causal attribution for defenses
under mechanism separability, and degenerates to a descriptive marginal
effect metric for coupled-intervention defenses
(see Section~\ref{sec:discussion}).

In the Tor/Wang empirical validation, the above theoretical structure is
quantitatively confirmed at the component granularity.
FRONT suppresses the direction chain component from~0.637 to~0.236~bits
($-62.9\%$), while the Markov residual (chain term~2) is almost unchanged
($0.018\to 0.021$~bits), with the 95\% CI $[0.016,0.027]$ of the
0.021~bits residual excluding zero---strictly consistent with the SCM
prediction of ``only the direction mechanism is replaced; the timing path
is not intervened upon.''
These numbers provide a reproducible quantitative baseline validating the
computational closure of the framework, rather than revealing new defense
vulnerabilities.

Several substantial open problems remain for landing and generalization of
the theoretical framework.
First, the explicit construction of the mapping $\psi:\mathcal{D}\to\Theta$
---i.e., how to estimate the corresponding center parameter
$\bar\theta^{(D)}$ from real defense parameters (e.g., FRONT's injection
window $w$)---is a necessary prerequisite for applying the trace formula to
defense parameter gradient optimization.
Second, SCM modeling of memory-bearing defenses (differential-privacy
noise injection, adaptive shaping) requires introducing temporal dependence
structures into the causal graph; the current stationary memoryless defense
assumption (Definition~\ref{def:memoryless-defense}) does not directly apply.
Third, mechanism separability does not strictly hold for splitting-type
defenses; defining cost functionals matched to the mechanisms of TrafficSliver
and similar defenses is a prerequisite for extending the causal efficacy
metric.
Finally, on variable-length packet protocols such as QUIC and HTTP/3,
$Y_{\mathrm{size}}$ no longer degenerates, and the dominant relationship
among chain components may be fundamentally different from the Tor scenario;
verifying whether direction dominance is a phenomenon specific to the Tor
protocol stack is a necessary empirical step for extending the framework
to general traffic analysis scenarios.

\section*{Acknowledgments}

The authors thank the anonymous reviewers for their constructive comments.
This work was supported in part by the National Natural Science Foundation
of China.

\bibliographystyle{IEEEtran}
\bibliography{IEEEabrv,TDSC}

\begin{thebibliography}{10}
\providecommand{\url}[1]{#1}
\csname url@rmstyle\endcsname
\providecommand{\newblock}{\relax}
\providecommand{\bibinfo}[2]{#2}
\providecommand\BIBentrySTDinterwordspacing{\spaceskip=0pt\relax}
\providecommand\BIBentryALTinterwordstretchfactor{4}
\providecommand\BIBentryALTinterwordspacing{\spaceskip=\fontdimen2\font plus
\BIBentryALTinterwordstretchfactor\fontdimen3\font minus
  \fontdimen4\font\relax}
\providecommand\BIBforeignlanguage[2]{{%
\expandafter\ifx\csname l@#1\endcsname\relax
\typeout{** WARNING: IEEEtran.bst: No hyphenation pattern has been}%
\typeout{** loaded for the language `#1'. Using the pattern for}%
\typeout{** the default language instead.}%
\else
\language=\csname l@#1\endcsname
\fi
#2}}

\bibitem{wang2014effective}
T.~Wang, X.~Cai, R.~Nithyanand, R.~Johnson, and I.~Goldberg, ``Effective
  attacks and provable defenses for website fingerprinting,'' in \emph{Proc.
  23rd {USENIX} Security Symp.}\hskip 1em plus 0.5em minus 0.4em\relax {USENIX}
  Association, 2014, pp. 143--157.

\bibitem{sirinam2018deep}
P.~Sirinam, M.~Imani, M.~Juarez, and M.~Wright, ``Deep fingerprinting:
  Undermining website fingerprinting defenses with deep learning,'' in
  \emph{Proc. 2018 ACM SIGSAC Conf. Comput. Commun. Security}.\hskip 1em plus
  0.5em minus 0.4em\relax ACM, 2018, pp. 1928--1943.

\bibitem{shen2021graphdapp}
M.~Shen, J.~Zhang, L.~Zhu, K.~Xu, X.~Du, and Y.~Liu, ``Accurate decentralized
  application identification via encrypted traffic analysis using graph neural
  networks,'' \emph{IEEE Trans. Inf. Forensics Security}, vol.~16, pp.
  2367--2380, 2021.

\bibitem{lin2022etbert}
X.~Lin, G.~Xiong, G.~Gou, Z.~Li, J.~Shi, and J.~Yu, ``{ET-BERT}: A
  contextualized datagram representation with pre-training transformers for
  encrypted traffic classification,'' in \emph{Proc. ACM Web Conf. 2022}.\hskip
  1em plus 0.5em minus 0.4em\relax ACM, 2022, pp. 633--642.

\bibitem{liu2026inevitability}
G.~Liu, G.~Cheng, and W.~Liu, ``The inevitability of side-channel leakage in
  encrypted traffic,'' \emph{arXiv preprint arXiv:2602.14055}, 2026.

\bibitem{mathews2023sok}
N.~Mathews, J.~K. Holland, S.~E. Oh, M.~S. Rahman, N.~Hopper, and M.~Wright,
  ``{SoK}: A critical evaluation of efficient website fingerprinting
  defenses,'' in \emph{Proc. 44th IEEE Symp. Security Privacy}.\hskip 1em plus
  0.5em minus 0.4em\relax IEEE, 2023, pp. 969--986.

\bibitem{pearl2009causality}
J.~Pearl, \emph{Causality: Models, Reasoning, and Inference}, 2nd~ed.\hskip 1em
  plus 0.5em minus 0.4em\relax Cambridge, UK: Cambridge Univ. Press, 2009.

\bibitem{hayes2016kfingerprinting}
J.~Hayes and G.~Danezis, ``k-fingerprinting: A robust scalable website
  fingerprinting technique,'' in \emph{Proc. 25th {USENIX} Security
  Symp.}\hskip 1em plus 0.5em minus 0.4em\relax {USENIX} Association, 2016, pp.
  1187--1203.

\bibitem{shen2023subverting}
M.~Shen, K.~Liu, L.~Zhu, \emph{et~al.}, ``Subverting website fingerprinting
  defenses with robust traffic representation,'' in \emph{Proc. 32nd {USENIX}
  Security Symp.}\hskip 1em plus 0.5em minus 0.4em\relax {USENIX} Association,
  2023.

\bibitem{reed2017netflix}
A.~Reed and M.~Kranch, ``Identifying {HTTPS}-protected {Netflix} videos in
  real-time,'' in \emph{Proc. 7th ACM Conf. Data Appl. Security Privacy}.\hskip
  1em plus 0.5em minus 0.4em\relax ACM, 2017, pp. 361--368.

\bibitem{qu2023hierarchical}
J.~Qu, X.~Ma, J.~Li, X.~Luo, L.~Xue, J.~Zhang, Z.~Li, L.~Feng, and X.~Guan,
  ``An input-agnostic hierarchical deep learning framework for traffic
  fingerprinting,'' in \emph{Proc. 32nd {USENIX} Security Symp.}\hskip 1em plus
  0.5em minus 0.4em\relax {USENIX} Association, 2023.

\bibitem{jin2023transformer}
Z.~Jin, M.~Shen, L.~Zhu, \emph{et~al.}, ``Transformer-based model for multi-tab
  website fingerprinting attack,'' in \emph{Proc. 2023 ACM SIGSAC Conf. Comput.
  Commun. Security}.\hskip 1em plus 0.5em minus 0.4em\relax ACM, 2023, pp.
  1050--1064.

\bibitem{deng2024earlystage}
X.~Deng, Q.~Li, and K.~Xu, ``Robust and reliable early-stage website
  fingerprinting attacks via spatial-temporal distribution analysis,'' in
  \emph{Proc. 2024 ACM SIGSAC Conf. Comput. Commun. Security}.\hskip 1em plus
  0.5em minus 0.4em\relax ACM, 2024.

\bibitem{cherubin2022online}
G.~Cherubin, R.~Jansen, and C.~Troncoso, ``Online website fingerprinting:
  Evaluating website fingerprinting attacks on {Tor} in the real world,'' in
  \emph{Proc. 31st {USENIX} Security Symp.}\hskip 1em plus 0.5em minus
  0.4em\relax {USENIX} Association, 2022, pp. 753--770.

\bibitem{siby2023quic}
S.~Siby, L.~Barman, C.~A. Wood, M.~Fayed, N.~Sullivan, and C.~Troncoso,
  ``Evaluating practical {QUIC} website fingerprinting defenses for the
  masses,'' \emph{Proc. Privacy Enhancing Technol.}, vol. 2023, no.~4, pp.
  79--95, 2023.

\bibitem{mei2025high}
H.~Mei, G.~Cheng, and Y.~Yuan, ``High precision and efficient anonymous traffic
  classification in the real-world,'' \emph{IEEE/ACM Trans. Netw.}, vol.~33,
  no.~3, pp. 966--981, 2025.

\bibitem{dyer2012peek}
K.~P. Dyer, S.~E. Coull, T.~Ristenpart, and T.~Shrimpton, ``Peek-a-boo, {I}
  still see you: Why efficient traffic analysis countermeasures fail,'' in
  \emph{Proc. 2012 IEEE Symp. Security Privacy}.\hskip 1em plus 0.5em minus
  0.4em\relax IEEE, 2012, pp. 332--346.

\bibitem{cai2014systematic}
X.~Cai, R.~Nithyanand, T.~Wang, R.~Johnson, and I.~Goldberg, ``A systematic
  approach to developing and evaluating website fingerprinting defenses,'' in
  \emph{Proc. 2014 ACM SIGSAC Conf. Comput. Commun. Security}.\hskip 1em plus
  0.5em minus 0.4em\relax ACM, 2014, pp. 227--238.

\bibitem{wang2017walkie}
T.~Wang and I.~Goldberg, ``Walkie-talkie: An efficient defense against passive
  website fingerprinting attacks,'' in \emph{Proc. 26th {USENIX} Security
  Symp.}\hskip 1em plus 0.5em minus 0.4em\relax {USENIX} Association, 2017, pp.
  1375--1390.

\bibitem{juarez2016wtfpad}
M.~Juarez, M.~Imani, M.~Perry, C.~Diaz, and M.~Wright, ``Toward an efficient
  website fingerprinting defense,'' in \emph{Computer Security -- {ESORICS}
  2016}.\hskip 1em plus 0.5em minus 0.4em\relax Springer, 2016, pp. 27--46.

\bibitem{sabzi2024netshaper}
A.~Sabzi, R.~Vora, S.~Goswami, M.~Seltzer, M.~L{\'e}cuyer, and A.~Mehta,
  ``{NetShaper}: A differentially private network side-channel mitigation
  system,'' in \emph{Proc. 33rd {USENIX} Security Symp.}\hskip 1em plus 0.5em
  minus 0.4em\relax {USENIX} Association, 2024, pp. 3385--3402.

\bibitem{shen2024palette}
M.~Shen, K.~Ji, J.~Wu, Q.~Li, X.~Kong, K.~Xu, and L.~Zhu, ``Real-time website
  fingerprinting defense via traffic cluster anonymization,'' in \emph{Proc.
  2024 IEEE Symp. Security Privacy}.\hskip 1em plus 0.5em minus 0.4em\relax
  IEEE, 2024, pp. 3238--3256.

\bibitem{huang2025wfa2d}
J.~Huang, W.~Liu, G.~Liu, B.~Gao, and F.~Nie, ``{WF-A2D}: Enhancing privacy
  with asymmetric adversarial defense against website fingerprinting,''
  \emph{IEEE Trans. Inf. Forensics Security}, vol.~20, pp. 4739--4754, 2025.

\bibitem{amari2016information}
S.-i. Amari, \emph{Information Geometry and Its Applications}.\hskip 1em plus
  0.5em minus 0.4em\relax Tokyo, Japan: Springer, 2016.

\bibitem{amari1998natural}
------, ``Natural gradient works efficiently in learning,'' \emph{Neural
  Comput.}, vol.~10, no.~2, pp. 251--276, 1998.

\bibitem{amari1995information}
------, ``Information geometry of the {EM} and em algorithms for neural
  networks,'' \emph{Neural Netw.}, vol.~8, pp. 1379--1408, 1995.

\bibitem{cover2006elements}
T.~M. Cover and J.~A. Thomas, \emph{Elements of Information Theory},
  2nd~ed.\hskip 1em plus 0.5em minus 0.4em\relax Hoboken, NJ: Wiley, 2006.

\bibitem{gong2020zero}
J.~Gong and T.~Wang, ``Zero-delay lightweight defenses against website
  fingerprinting,'' in \emph{Proc. 29th {USENIX} Security Symp.}\hskip 1em plus
  0.5em minus 0.4em\relax {USENIX} Association, 2020, pp. 717--734.

\bibitem{delacadena2020trafficsliver}
W.~De~la Cadena, A.~Mitseva, J.~Hiller, J.~Pennekamp, S.~Reuter, J.~Filter,
  T.~Engel, K.~Wehrle, and A.~Panchenko, ``{TrafficSliver}: Fighting website
  fingerprinting attacks with traffic splitting,'' in \emph{Proc. 2020 ACM
  SIGSAC Conf. Comput. Commun. Security}.\hskip 1em plus 0.5em minus
  0.4em\relax ACM, 2020, pp. 1971--1985.

\end{thebibliography}

\end{document}